\documentclass[a4paper,UKenglish]{lipics-v2018}
\usepackage{microtype}%if unwanted, comment out or use option "draft"
\usepackage{algorithm}
\usepackage{multicol}
\usepackage[noend]{algpseudocode}

\algnewcommand\algorithmicforeach{\textbf{for each}}
\algdef{S}[FOR]{ForEach}[1]{\algorithmicforeach\ #1\ \algorithmicdo}

\newcommand{\floor}[1]{\lfloor #1 \rfloor}

\newcommand{\dsh}{{\delta_{sh}}}
\newcommand{\alg}{{a_{lg}^*}}
\newcommand{\ds}{\displaystyle}

\newcommand{\longshort}[2]{#1} %S~<DEL
\title{A Reallocation Algorithm for Online Split Packing of Circles}
\titlerunning{A Reallocation Algorithm for Online Split Packing of Circles} %optional, please use if title is longer than one line

\author{Shunhao Oh}{School of Computing, National University of Singapore, Singapore}{ohoh@u.nus.edu}{}{}%mandatory, please use full name; only 1 author per \author macro; first two parameters are mandatory, other parameters can be empty.

\author{Seth Gilbert}{School of Computing, National University of Singapore, Singapore}{seth.gilbert@comp.nus.edu.sg}{}{}%mandatory, please use full name; only 1 author per \author macro; first two parameters are mandatory, other parameters can be empty.

\authorrunning{S. Oh and S. Gilbert} %mandatory. First: Use abbreviated first/middle names. Second (only in severe cases): Use first author plus 'et. al.'

\Copyright{Shunhao Oh and Seth Gilbert}%mandatory, please use full first names. LIPIcs license is "CC-BY";  http://creativecommons.org/licenses/by/3.0/

\subjclass{
\ccsdesc[500]{Theory of computation~Packing and covering problems},
\ccsdesc[500]{Theory of computation~Online algorithms},
\ccsdesc[500]{Theory of computation~Computational geometry}
}% mandatory: Please choose ACM 2012 classifications from https://www.acm.org/publications/class-2012 or https://dl.acm.org/ccs/ccs_flat.cfm . E.g., cite as "General and reference $\rightarrow$ General literature" or \ccsdesc[100]{General and reference~General literature}.
\keywords{circle packing, online algorithms, dynamic resource allocation}% mandatory: Please provide 1-5 keywords
\EventEditors{John Q. Open and Joan R. Acces}
\EventNoEds{2}
\EventLongTitle{42nd Conference on Very Lovely Circles (CIRCLELOVERS)}
\EventShortTitle{CIRCLELOVERS}
\EventAcronym{CIRCLELOVERS}
\EventYear{2018}
\EventDate{December 24--27, 2018}
\EventLocation{Little Whinging, United Kingdom}
\EventLogo{}
\SeriesVolume{42}
\ArticleNo{23}
\nolinenumbers %uncomment to disable line numbering
\hideLIPIcs  %uncomment to remove references to LIPIcs series (logo, DOI, ...), e.g. when preparing a pre-final version to be uploaded to arXiv or another public repository
\begin{document}

\maketitle

\begin{abstract}
The Split Packing algorithm \cite{splitpacking_ws, splitpackingsoda, splitpacking} is an offline algorithm that packs a set of circles into triangles and squares up to critical density. In this paper, we develop an online alternative to Split Packing to handle an online sequence of insertions and deletions, where the algorithm is allowed to reallocate circles into new positions at a cost proportional to their areas. The algorithm can be used to pack circles into squares and right angled triangles. If only insertions are considered, our algorithm is also able to pack to critical density, with an amortised reallocation cost of $O(c\log \frac{1}{c})$ for squares, and $O(c(1+s^2)\log_{1+s^2}\frac{1}{c})$ for right angled triangles, where $s$ is the ratio of the lengths of the second shortest side to the shortest side of the triangle, when inserting a circle of area $c$. When insertions and deletions are considered, we achieve a packing density of $(1-\epsilon)$ of the critical density, where $\epsilon>0$ can be made arbitrarily small, with an amortised reallocation cost of $O(c(1+s^2)\log_{1+s^2}\frac{1}{c} + c\frac{1}{\epsilon})$.
 \end{abstract}

%S~>INS \newpage % Title page page break

\section{Introduction}
\label{section:introduction}
A common class of problems in data structures requires handling a sequence of online requests. In a dynamic resource allocation problem, one has to handle a sequence of allocation and deallocation requests. A good allocation algorithm would run fast, and allocate items in a way that uses as few resources as possible to store the items.

Classical solutions disallow the algorithm from reallocating already-placed items except during deletion. In our problem, the algorithm is allowed to reallocate already-placed items with additional cost. Instead of bounding the running time of the algorithm, we focus on trying to bound the reallocation costs required to handle the sequence of requests, while minimising the amount of space required to handle all the requests.

% XXX Toy real world example

We look at online circle packing, where we try to dynamically pack a set of circles of unequal areas into a unit square while allowing reallocations. Our work builds on insights from the Split Packing papers by Fekete, Morr, and Scheffer \cite{splitpacking_ws, splitpackingsoda, splitpacking}. The Split Packing papers derive the critical density $a$ of squares and obtuse triangles (triangles with an internal angle of at least $90^\circ$). The critical of density a region is the largest value $a$ such that any set of circles with total area at most $a$ can be packed into the region. To simplify our discussion, we sometimes refer to the critical density of a region as the ``capacity'' of the region.

As shown in the paper \cite{splitpackingsoda, splitpacking}, the critical density for a square is equal to the combined area of the two circles in the configuration in Figure \ref{fig:splitsquare_before}. The critical density for an obtuse triangle is the area of its incircle. The Split Packing Algorithm, presented in the paper, is an offline algorithm that packs circles into squares and obtuse triangles to critical density.
Figure \ref{fig:split_packing_examples} shows example circle packings produced by the Split Packing Algorithm and our Online Split Packing Algorithm.

\begin{figure}[!h]
  \centering
    \begin{subfigure}[b]{.4\linewidth}
      \centering
      \includegraphics[width=.6\linewidth]{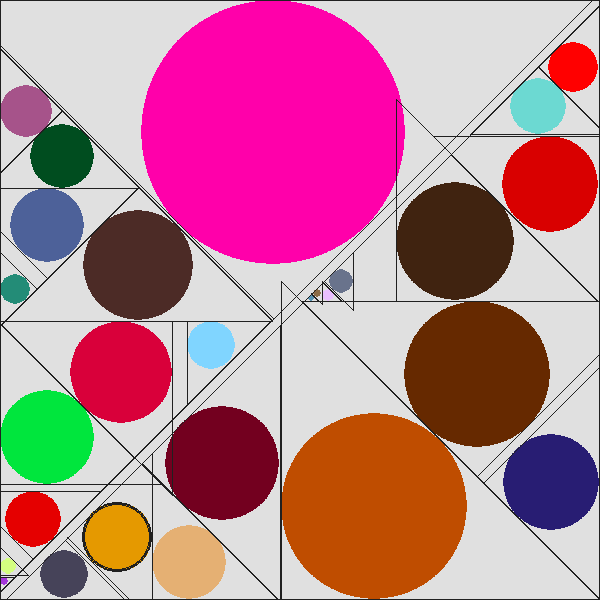}
      \caption{Split Packing}
      \label{fig:offline_split_packing}
    \end{subfigure}%
    \begin{subfigure}[b]{.4\linewidth}
      \centering
      \includegraphics[width=.6\linewidth]{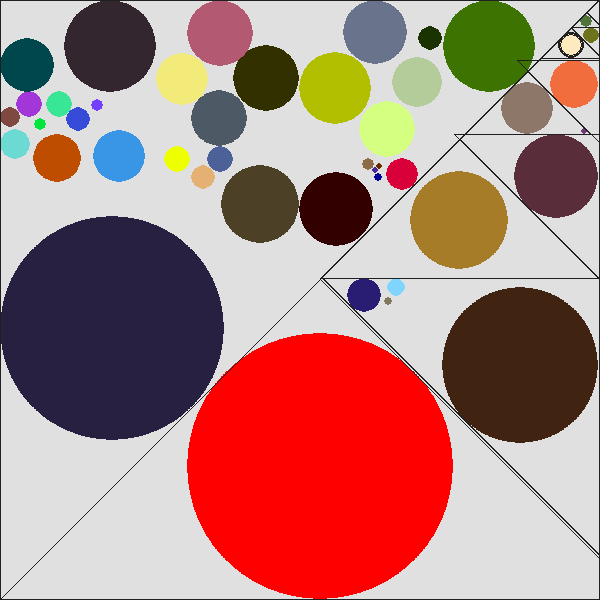}
      \caption{Online Split Packing}
      \label{fig:online_split_packing}
    \end{subfigure}%
  \caption{Sample packings produced by the two algorithms.}
  \label{fig:split_packing_examples}
\end{figure}

\subsection{The Problem of Online Circle Packing}
% CHECKED
In online circle packing, we are given a fixed region, and an online sequence of insertion and deletion requests for circles of varying sizes into the region. In between requests, all circles need to be packed within the bounds of the region such that no two circles overlap. While traditionally circles are not to be moved once placed, in our setting, at any point of time, the algorithm is allowed to reallocate sets of circles, at a cost proportional to the sum of the areas of the circles reallocated (volume cost). We do not require reallocations to be in place, meaning that when a set of circles is to be reallocated to new positions, we can see them as being simultaneously removed from the packing, then placed in their new locations. (This is in contrast to requiring circles to be moved one at a time to their new locations.) We aim to bound the total reallocation cost incurred by the algorithm throughout the packing.

An allocation algorithm is said to have a packing density of $A$ on a region if it can handle any sequence of insertion and deletion requests into the region, as long as the total area of the circles in the region at any point of time is at most $A$. The algorithm packs to critical density on a region if it attains a packing density equal to the region's capacity.

%To obtain good bounds on the reallocation costs, one may need to pack to less than critical density to allow for sufficient space to move circles.

\subsection{Results}
Our main result is the Online Split Packing Algorithm, an dynamic circle packing algorithm based on the original Split Packing Algorithm. The Online Split Packing Algorithm can be used to pack circles into squares and right angled triangles.
\begin{enumerate}
\item
For insertions only into a square, the algorithm packs to critical density, with an amortised reallocation cost of $O(c(\log_2 \frac{1}{c}))$ for inserting a circle of area $c$.

\item
For insertions only into a right angled triangle of side lengths $\ell$, $s\ell$ and $\ell\sqrt{1+s^2}$, where $s \geq 1$, the algorithm can achieve critical density, with an amortised reallocation cost of $O(c(1+s^2)\log_{1+s^2}\frac{1}{c})$ for inserting a circle of area $c$.

\item
For insertions and deletions into a region (a square or a right-angled triangle) of capacity $a$, the algorithm achieves a packing density of $a(1-\epsilon)$, where $\epsilon$ can be defined to be arbitrarily small. Consequently, we add an amortised reallocation cost of $O(c\frac{1}{\epsilon})$ for inserting a circle of area $c$. In other words, we achieve an amortised reallocation cost of $O(c(\log_2 \frac{1}{c} + \frac{1}{\epsilon}))$ for squares, and $O(c((1+s^2)\log_{1+s^2}\frac{1}{c} + \frac{1}{\epsilon}))$ for right angled triangles.
\end{enumerate}

\subsection{Challenges of Online Circle Packing}

The problem of packing circles into a square is a difficult problem in general. The decision problem of whether a set of circles of possibly unequal areas can be packed into a square has been shown to be NP-hard \cite{demaine2010circle}. Given $n$ equal circles, the problem of finding the smallest square that can fit these circles only has proven optimal solutions for $n \leq 35$ \cite{locatelli2002packing}. Packing a square to critical density has only been recently solved with Split Packing in 2017.

The best known algorithm for online packing of squares into squares is given by Brubach \cite{improvedsquareintosquare}, with a packing density of $\frac{2}{5}$. By embedding circles in squares, we obtain an online circle packing algorithm with a packing density of $\frac{\pi}{4}\times\frac{2}{5} \approx 0.3142$ \cite{splitpacking}. The critical density for offline packing of circles into squares, as given by Split Packing, is $\frac{\pi}{3+2\sqrt{2}} \approx 0.5390$.

The required arrangement of circles in a tight packing is highly dependent on the distribution of circle sizes. The Split Packing Algorithm \cite{splitpackingsoda, splitpacking} packs to critical density by packing circles in a top-down, divide and conquer fashion that starts with sorting and partitioning the circles by size. Small changes to the inputs may thus result in very different packings. An example of this is shown in Figure \ref{fig:packvis}.

\begin{figure}[!h]
  \centering
    \begin{subfigure}[b]{.3\linewidth}
      \centering
      \includegraphics[width=.7\linewidth]{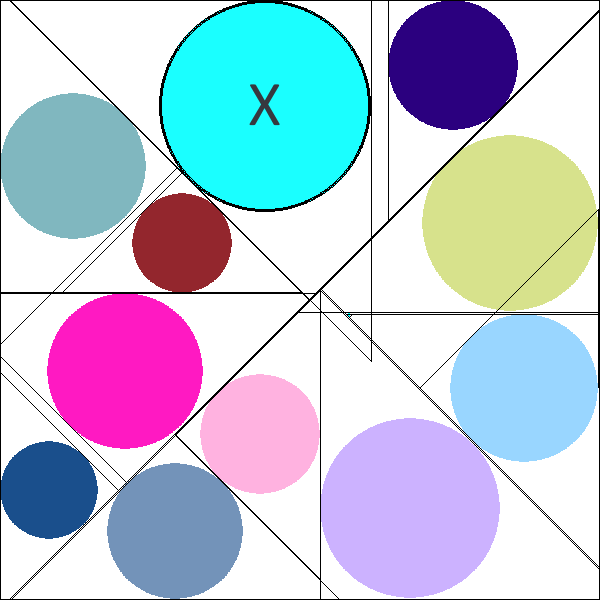}
      \label{fig:packvis_size1}
    \end{subfigure}%
    \begin{subfigure}[b]{.3\linewidth}
      \centering
      \includegraphics[width=.7\linewidth]{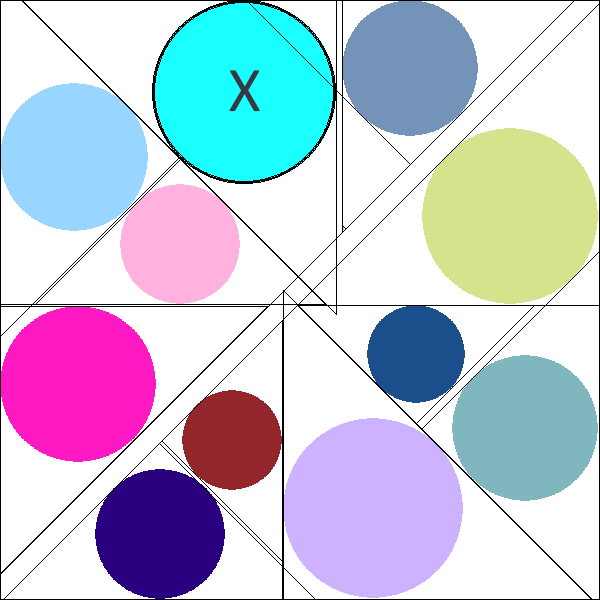}
      \label{fig:packvis_size2}
    \end{subfigure}%
    \begin{subfigure}[b]{.3\linewidth}
      \centering
      \includegraphics[width=.7\linewidth]{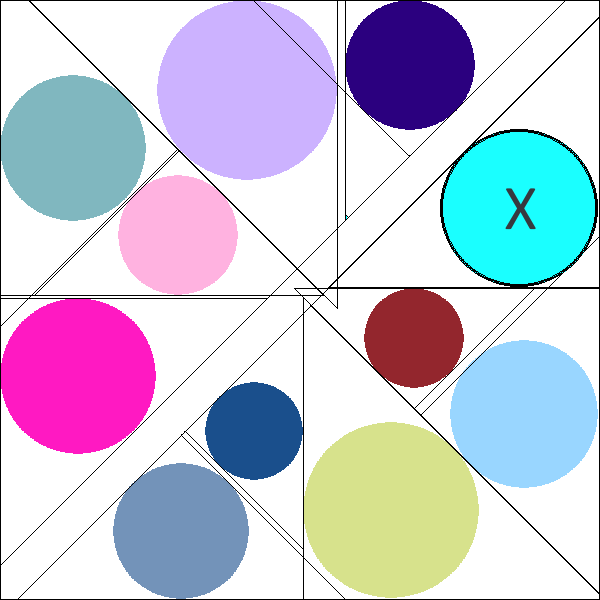}
      \label{fig:packvis_size3}
    \end{subfigure}
  \caption{Split Packing Algorithm differences when the size of the circle $X$ is altered slightly.}
  \label{fig:packvis}
\end{figure}

%With insertions and deletions, larger circles may be replaced with smaller circles within a subregion of the original shape, which may cause problems if the subregion can no longer support that set of circles. 

%When packing with reallocations there is a need to be able to repack only a part of the configuration at a time. With insertions and deletions, larger circles may be replaced with smaller circles within a subregion of the original shape, which may cause problems as if the subregion can no longer support that set of circles. 

Our algorithm packs an online sequence of circles tightly by allowing reallocations. When packing with reallocations, we aim to repack only a part of the configuration at a time. We do this by making use of a binary tree structure similar to the original Split Packing Algorithm. By designing a new packing strategy, we maintain a tree where left children are packed tightly, while extra slack is kept in the rightmost node of each level (this refers to nodes on the right spine of the binary tree) to allow for circle movement. New circles are recursively inserted into the rightmost node of each level, where it eventually triggers a repack of an entire subtree at a certain level to allocate the circle.

\subsection{Related Work}
There are many examples of existing work in online resource allocation which do not allow items to be reallocated once placed. Online square packing into squares has been studied by Han et al. \cite{onlineremovablesquarepacking} and Fekete and Hoffman \cite{square_into_square}. Other examples are squares in unbounded strips \cite{onlinesquarepacking}, rectangle packing \cite{ontwodimensionalpacking} and multiple variants of online bin packing \cite{improvedsquareintosquare, online_multidim_bin_packing, Epstein2005, onlinecubes, onespaceboundedtwodimensional}.

%While resource allocation problems that allow for reallocations are less commonly studied, there is existing work that employs similar models to
Allowing reallocations can lead to better results than if reallocations were not permitted. Ivkovi{\'c} and Lloyd \cite{ivkovic2009fully} provide an algorithm for bin packing that is $\frac{5}{4}$-competitive with the best practical offline algorithms by allowing reallocations. This beats the lower bound of $\frac{4}{3}$ (proven in the same paper) if reallocations are not considered.
Berndt et al. \cite{berndt2014fully} study this problem further and achieves better upper and lower bounds.
There is also work that aims to bound reallocation costs rather than running times. Fekete et al. \cite{fekete2017efficient} study online square packing with reallocations. Bender et al. study dynamic resource allocation with reallocation costs in problems like scheduling \cite{bender2015reallocation}, memory allocation \cite{bender2014cost} and maintaining modules on an FPGA \cite{maintainingarrays}.
We also take ideas from packed memory arrays \cite{cacheobliviousbtrees, adaptivepma}.

Most existing work on circle packing focuses on the global, offline optimization problem of packing either equal or unequal circles into various shapes. Many existing results are collected in the packomania website \cite{packomania}. Many heuristic methods have also been developed for offline circle packing. An example is \cite{george1995packing}, which explores a mix of different strategies, like nonlinear mixed integer programming and genetic algorithms to fit circles of unequal sizes into a rectangle. Offline packings of equal-size spheres into a cube are looked into in \cite{spherepacking}. We refer the reader to \cite{circleliteraturereview} for a review of other circle and sphere packing methods. 

Recent work on online circle packing focuses on packing an online sequence of circles into the minimum number of square bins. Hokama et al. \cite{hokama2016bounded} provide an asymptotic competitive ratio of $2.4394$, and gives a lower bound of $2.2920$ for the problem. The upper bound was improved to $2.3536$ in \cite{lintzmayer2017online}. We have not found existing work that takes into account the possibility of reallocations for online circle packing.

\section{Background - The Split Packing Algorithm}
\label{section:background}
% CHECKED
We briefly describe the original Split Packing Algorithm. Given a square or obtuse triangle as the region, and any set of circles with total area at most the capacity of the region, the Split Packing Algorithm can pack the set of circles into the region.

The algorithm works in a top down, divide-and-conquer fashion. In the simple case, if the Split Packing Algorithm is given a single circle to be packed into a region, the circle is simply placed into the center of the region. If given at least two circles to be packed, the algorithm partitions the circles into two sets, splits the region into two smaller subregions, and recursively packs each set of circles into its corresponding subregion. Using a binary tree analogy for this algorithm, these subregions are the two children of the original region.

\begin{figure}[!h]
  \centering
    \begin{subfigure}[b]{.27\linewidth}
      \centering
      \includegraphics[width=.93\linewidth]{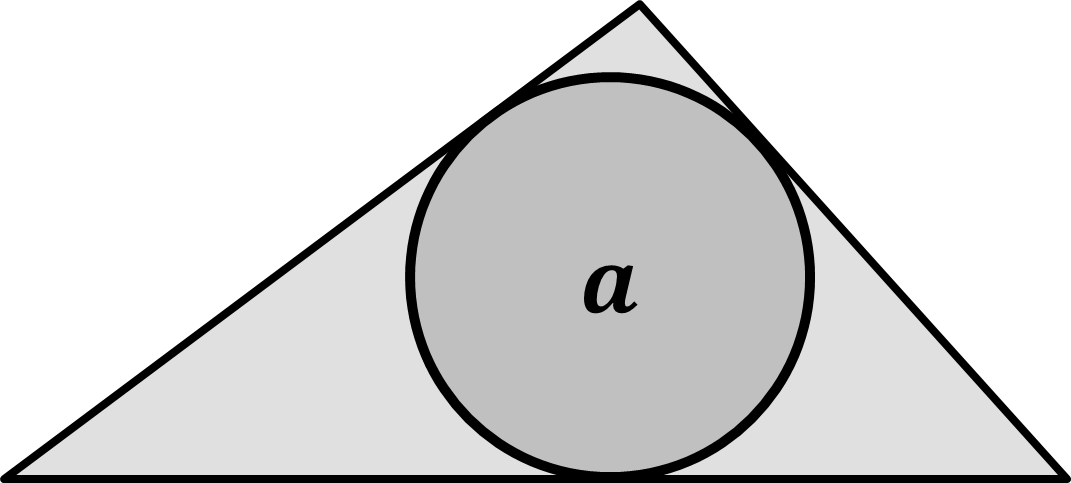}
      \caption{}
      \label{fig:splittri_before}
    \end{subfigure}%
    \begin{subfigure}[b]{.27\linewidth}
      \centering
      \includegraphics[width=.93\linewidth]{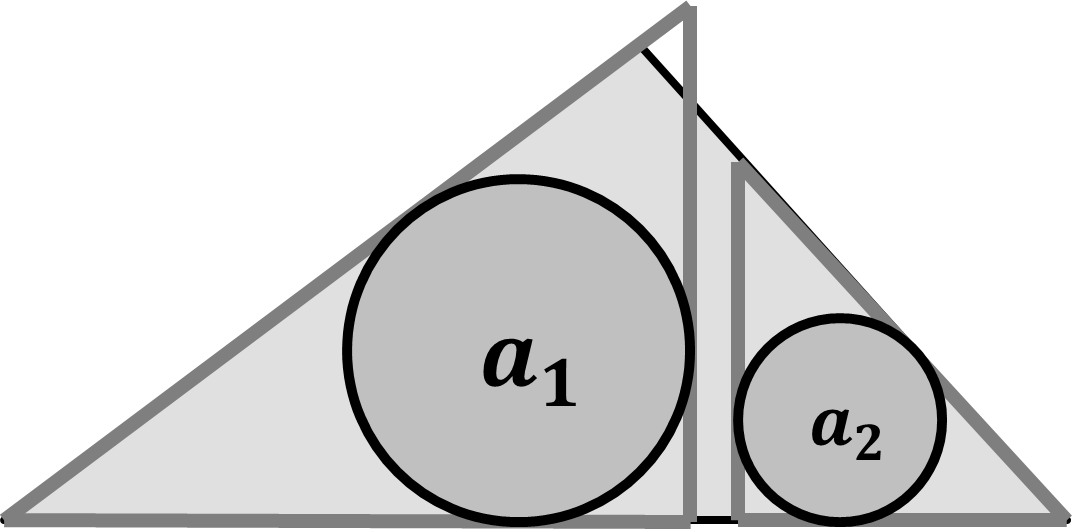}
      \caption{}
      \label{fig:splittri_after}
    \end{subfigure}%
    \begin{subfigure}[b]{.21\linewidth}
      \centering
      \includegraphics[width=.64\linewidth]{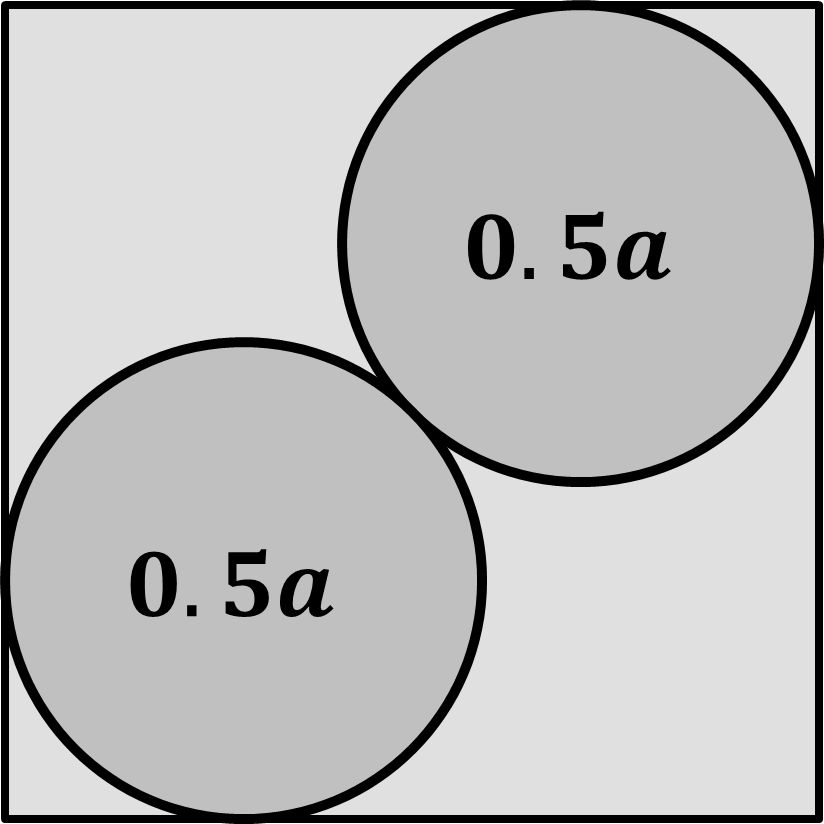}
      \caption{}
      \label{fig:splitsquare_before}
    \end{subfigure}%
    \begin{subfigure}[b]{.25\linewidth}
      \centering
      \includegraphics[width=.64\linewidth]{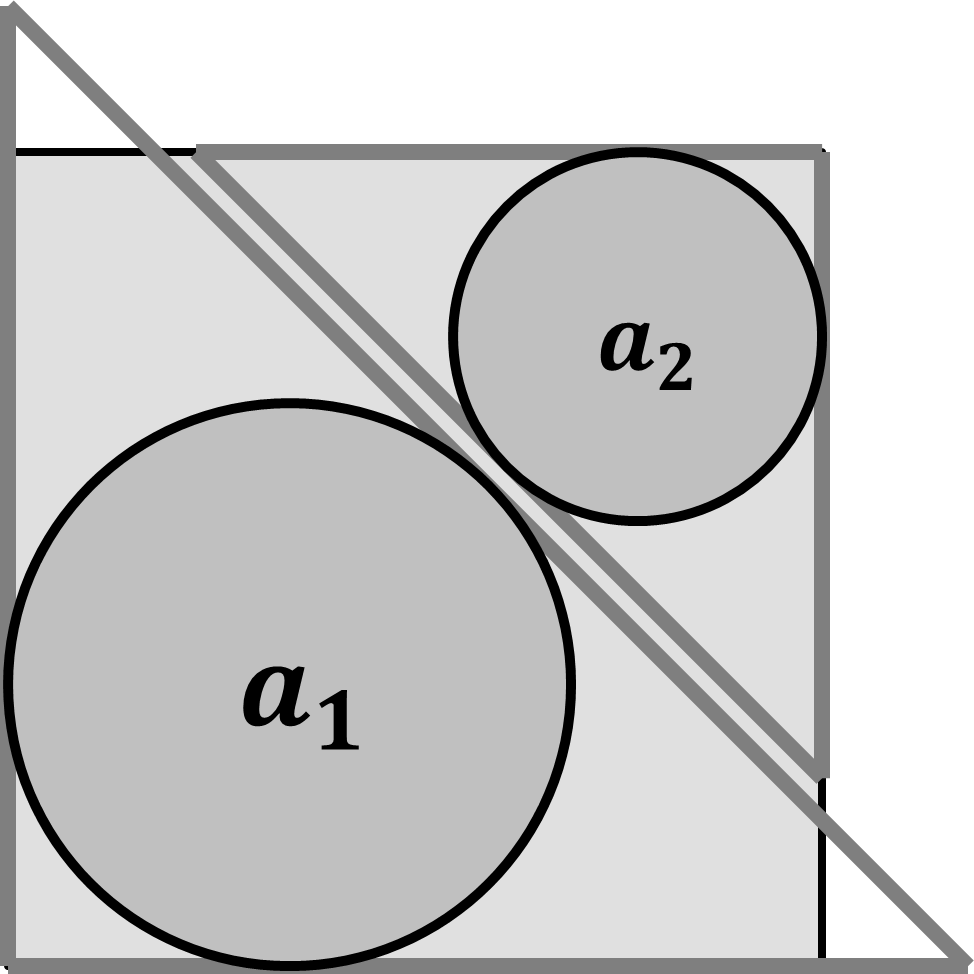}
      \caption{}
      \label{fig:splitsquare_after}
    \end{subfigure}%
  \caption{Splitting an obtuse triangle or square into two right angled triangles.}
  \label{fig:splittrisquare}
\end{figure}

% CHECKED
The Split Packing papers \cite{splitpackingsoda, splitpacking} describe how a region of capacity $a$ is to be split into two smaller regions of capacities $a_1$ and $a_2$ respectively, where $a_1 + a_2 \leq a$. If the original region is an obtuse triangle, $a$ is its incircle area. As shown in Figure \ref{fig:splittri_after}, the two subregions will be right angled triangles, defined by squeezing circles of areas $a_1$ and $a_2$ respectively into the left and right corners of the original triangle, which has been oriented such that the obtuse angle is at the top. If the original region is a square, $a$ is as shown in Figure \ref{fig:splitsquare_before}. As shown in Figure \ref{fig:splitsquare_after}, the two subregions will be isosceles right angled triangles, defined by squeezing circles of areas $a_1$ and $a_2$ respectively into opposite corners of the square.

In both cases, the capacities of the two subregions will then be $a_1$ and $a_2$ respectively. The Split Packing papers prove that as long as $a_1 + a_2 \leq a$, the two subregions will not overlap each other. Note that the two subregions may not be contained within the original region, as shown in Figure \ref{fig:splittrisquare}. To account for this, Split Packing rounds the corners of triangles to form hats, which will remain within the bounds of their parent regions.

To decide on the values of $a_1$ and $a_2$ for the split, the Split Packing Algorithm calls \Call{SPLIT}{$C$, $F$} (Algorithm \ref{alg:split}) to partition the set of circles $C$ into two sets, $C_1$ and $C_2$, to be packed into the left and right children respectively. $a_1$ and $a_2$ are then determined to be the sums of the areas of the circles in $C_1$ and $C_2$ respectively.

\Call{SPLIT}{$C$, $F$} has a second parameter, the ideal split key $F = (f_1,f_2)$, which is a function of the shape we are packing the circles into. Intuitively, the ideal split key is represents the ``ideal split'' of the shape into two shapes of smaller capacity. The ratio $f_1 : f_2$ would be the ratio of the capacities of the left and right children in this ideal split.

\begin{algorithm}
\caption{\Call{SPLIT}{}, a greedy algorithm to partition the set $C$}\label{alg:split}
\begin{algorithmic}[1]
\Procedure{Split}{$C,F=(f_1,f_2)$}
    \State $C_1 \gets \emptyset$
    \State $C_2 \gets \emptyset$
    \ForEach {$c \in C$ in decreasing order of area}
        \If {$\frac{\text{sum}(C_1)}{f_1} < \frac{\text{sum}(C_2)}{f_2}$} \Comment{sum($X$) is the total area of the circles in $X$}
    \State  $C_1 \gets C_1 \cup \{c\}$
        \Else
    \State  $C_2 \gets C_2 \cup \{c\}$
        \EndIf
    \EndFor
    \Return $(C_1, C_2)$
\EndProcedure
\end{algorithmic}
\end{algorithm}

As seen in the \Call{SPLIT}{$C$, $F$} algorithm, the largest circle in $C$ will always be placed into $C_1$, the child corresponding to $f_1$ in the split key. Thus, by swapping the two children if needed, we have the freedom to choose which of the two children the largest circle in $C$ will be packed into. We take advantage of this \longshort{in the proofs of Lemmas \ref{lem:splitpackingsemihats} and\ref{lem:repackcorrectnesstriangles} later on}{later to prove some properties of our algorithm}.

%For squares, the ideal split key is simply $1:1$. For obtuse triangles, we compute $(f_1,f_2)$ by drawing a single vertical line down from the corner with the obtuse angle to perpendicularly intersect the opposite edge, partitioning the triangle into two right angled triangles. $f_1$ and $f_2$ are then the incircle areas of these two triangles. The ideal split key for a hat is the ideal split key of its underlying triangle, before any rounding of the corners is done.
% ^ we only explain this to explain the symmetry bit.

The paper goes further to describe how the same method can be used to pack objects of similar shapes into the same regions. However, the Split Packing Algorithm described is strictly offline, and a method to do this Split Packing dynamically has not been explored.

\section{Definitions}

Similar to the original Split Packing Algorithm, the Online Split Packing Algorithm has a binary tree structure, where each region is recursively split into two smaller nonoverlapping subregions. The root node is the original region we are packing circles into, which is either a square or a right angled triangle. A square splits into two right angled triangles, and a right angled triangle splits into two smaller right angled triangles.

%%S~CROP_START - Separate Diagrams
\begin{figure}[!h]
  \centering
  \includegraphics[width=.35\linewidth]{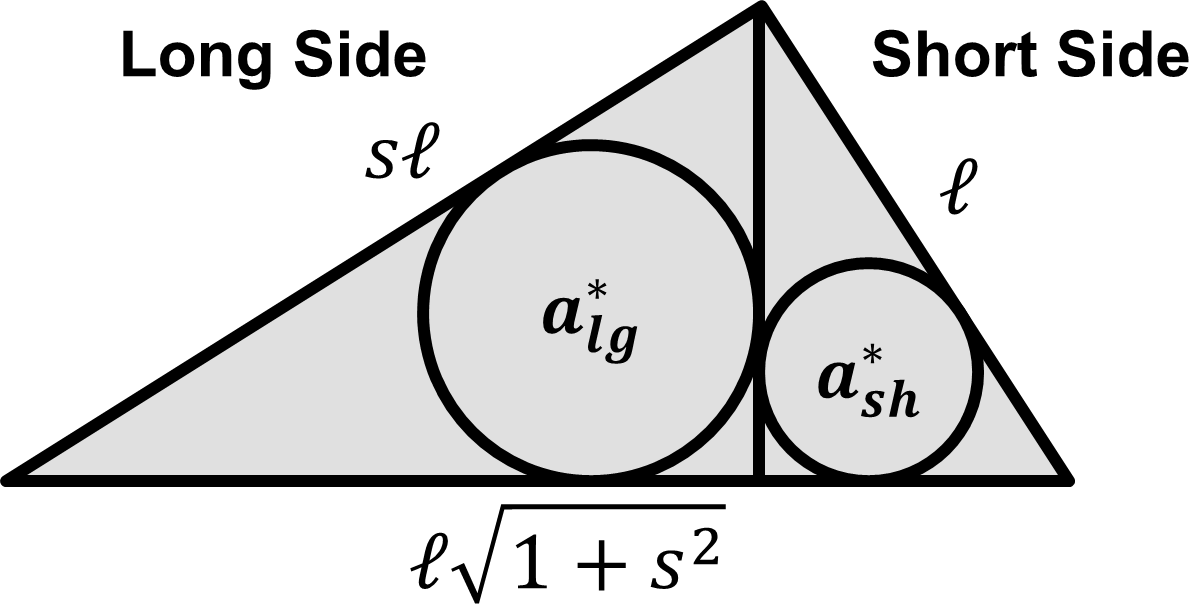}
  \caption{An $s$-right angled triangle, and its ideal split in to two smaller $s$-right angled triangles.}
  \label{fig:idealcapacities}
\end{figure}
%%S~CROP_END - Separate Diagrams

\begin{definition}[$s$-Right Angled Triangle]
For any $s \geq 1$, we refer to a right angled triangle with side lengths $\ell$, $s\ell$ and $\ell \sqrt{1+s^2}$ as an $s$-right angled triangle. $s$ is the ratio of the lengths of the two shortest sides. In Figure \ref{fig:idealcapacities}, a right-angled triangle is oriented such that the hypotenuse is at the bottom. We can split the triangle vertically into two sides. The long side is closer to the side with length $s\ell$, and the short side is closer to the side with length $\ell$.
\label{dfn:srightangledtriangle}
\end{definition}

\begin{definition}[$s$-shape]
To simplify our discussion of the amortised reallocation cost for both right angled triangles and squares later on, we define an $s$-shape to refer to:
\begin{enumerate}
\item A $1$-right angled triangle or a square if $s=1$
\item An $s$-right angled triangle if $s>1$
\end{enumerate}
\label{dfn:sshapes}
\end{definition}

Splitting an $s$-shape will form two more $s$-shapes, for the same value of $s$. Notably, a square is a $1$-shape, and all descendant nodes of a square will be $1$-right angled triangles.
%For the case of packing a square, the root node, a square, is a $1$-shape, and all descendants of the root node will be $1$-right angled triangles.

%%S~CROP_START - Separate Diagrams
\begin{figure}[!h]
  \centering
    \begin{subfigure}[b]{.35\linewidth}
      \centering
      \includegraphics[width=.6\linewidth]{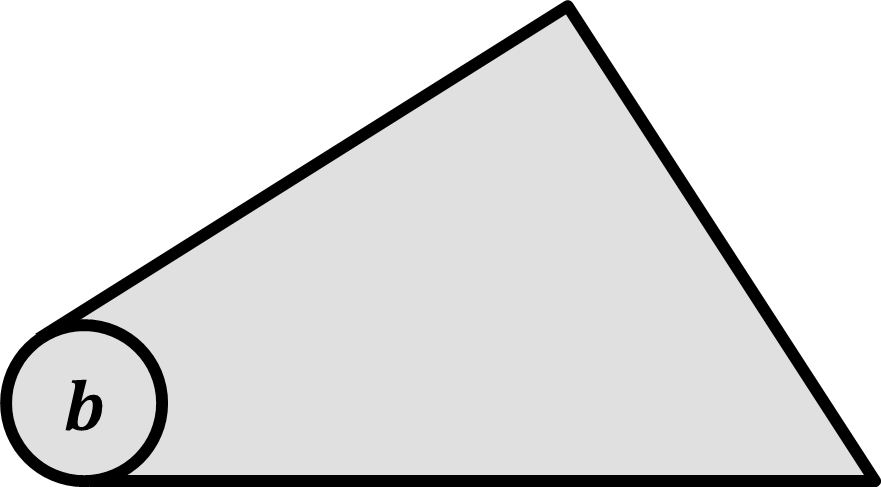}
      \caption{A $b$-semihat.}
      \label{fig:bsemihat}
    \end{subfigure}%
    \begin{subfigure}[b]{.35\linewidth}
      \centering
      \includegraphics[width=.6\linewidth]{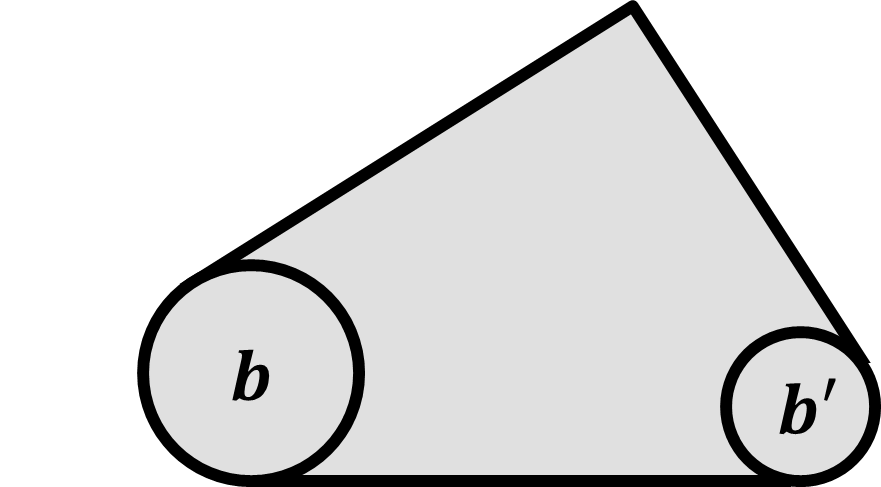}
      \caption{A $b$-$b'$-semihat.}
      \label{fig:bbsemihat}
    \end{subfigure}
  \caption{Semihats from a right angled triangle.}
  \label{fig:semihats}
\end{figure}
%%S~CROP_END - Separate Diagrams

%%S~>INS  \begin{figure}[!h]
%%S~>INS    \centering
%%S~>INS      \begin{subfigure}[b]{.33\linewidth}
%%S~>INS        \centering
%%S~>INS        \includegraphics[width=.9\linewidth]{diagrams/idealsplit.png}
%%S~>INS        \caption{An $s$-right angled triangle}
%%S~>INS        \label{fig:idealcapacities}
%%S~>INS      \end{subfigure}%
%%S~>INS      \begin{subfigure}[b]{.3\linewidth}
%%S~>INS        \centering
%%S~>INS        \includegraphics[width=.7\linewidth]{diagrams/bsemihat.png}
%%S~>INS        \caption{A $b$-semihat.}
%%S~>INS        \label{fig:bsemihat}
%%S~>INS      \end{subfigure}%
%%S~>INS      \begin{subfigure}[b]{.3\linewidth}
%%S~>INS        \centering
%%S~>INS        \includegraphics[width=.7\linewidth]{diagrams/bbsemihat.png}
%%S~>INS        \caption{A $b$-$b'$-semihat.}
%%S~>INS        \label{fig:bbsemihat}
%%S~>INS      \end{subfigure}
%%S~>INS  \caption{\textbf{(a)}: An $s$-right angled triangle, and its ideal split in to two smaller $s$-right angled triangles. \textbf{(b),(c)}: Semihats from a right angled triangle.}
%%S~>INS    \label{fig:semihats}
%%S~>INS  \end{figure}

\begin{definition}
We define the following terms:
\begin{enumerate}
\item \textbf{($b$-curve)}:
We say the long/short side (Definition \ref{dfn:srightangledtriangle}) of a triangle has a $b$-curve if its corresponding corner has been trimmed to an arc of a circle with area $b$, so that such a circle fits snugly into the corner. Figure \ref{fig:bsemihat} is a triangle with a $b$-curve on the long side.

\item \textbf{($b$-semihat)}:
A triangle with a $b$-curve on the long side.

\item \textbf{($b$-hat)}:
A triangle where both sides have a $b$-curve.

\item \textbf{($b$-$b'$-semihat)}:
The intersection between a $b'$-hat and a $b$-semihat. It is a triangle where the short side has a $b'$-curve, and the long side has a $\max\{b,b'\}$-curve.
A $b'$-hat can also be thought of as a $0$-$b'$-semihat.

\item \textbf{(Full triangle)}:
A $0$-hat. A full triangle has no rounded corners.

\item \textbf{(Capacity)}:
The \textbf{capacity} of a square or a right angled triangle is the largest value $a$ such that any set of circles with total area at most $a$ can be packed within its bounds. (This is the same as the critical density of its shape defined in Section \ref{section:introduction}). The capacity of a hat or semihat is the capacity of its underlying triangle.

\item \textbf{(Total Size)}:
The \textbf{total size} of a shape is the sum of the areas of the circles within the shape's bounds. We write $\Call{totalSize}{C}$ to represent the total area of a set of circles $C$.

\end{enumerate}
\end{definition}

In this paper, we use the term ``triangles'' to refer to full triangles, hats and semihats formed from right angled triangles. We also extend our definition of $s$-shapes (Definition \ref{dfn:sshapes}) to also refer to hats and semihats formed from $s$-right angled triangles.

We also introduce the concept of the ideal capacities $a_{lg}^*$ and $a_{sh}^*$, which represent the ``ideal split'' of a triangle into two smaller triangles. We make reference to these ideal capacities in many parts of our algorithm. Intuitively, we are concerned with the ideal splits as a triangle split by its ideal capacities (Figure \ref{fig:idealcapacities}) divides perfectly into two smaller triangles, with capacities adding up to the original triangle's capacity. In non-ideal splits, subregions may require rounding their corners to remain within the bounds of the original triangle, with the degree of rounding depending on how much their capacities deviate from the ideal.

\begin{definition}[Ideal Capacities of a Triangle]
Consider splitting a triangle by its split ratio as defined in \cite{splitpackingsoda, splitpacking}. This refers to drawing a vertical line from the right-angled corner to perpendicularly intersect the hypotenuse as shown in Figure \ref{fig:idealcapacities}. When an $s$-right angled triangle is split this way, the subregions on the long and short sides have capacities $a_{lg}^*$ and $a_{sh}^*$ respectively. We refer to these as their ideal capacities.
\label{dfn:idealcapacities}
\end{definition}

% SPLIT-----------

\section{The Online Split Packing Algorithm}
\label{section:onlinesplitpackingalgorithm}

This section describes the Online Split Packing Algorithm. Our algorithm relies on insights from the Split Packing Algorithm for packing circles into a square. We first discuss a version of the algorithm that only handles insertions. Deletions will be discussed in Section \ref{section:insertionsanddeletions}.

We make use of a binary tree structure to represent our packing. Each node of the binary tree represents a subregion of the original shape, with the root of the binary tree being the original shape we are packing circles into. Every non-root node in the tree is a subregion of its parent node. Nodes of the tree can contain circles. These contained circles are the circles we intend to pack within the bounds of the node. The circles contained in a node includes all circles contained within its descendants. A node is empty if it contains no circles.

\begin{figure}[!h]
  \centering
    \begin{subfigure}[b]{.5\linewidth}
      \centering
      \includegraphics[width=.8\linewidth]{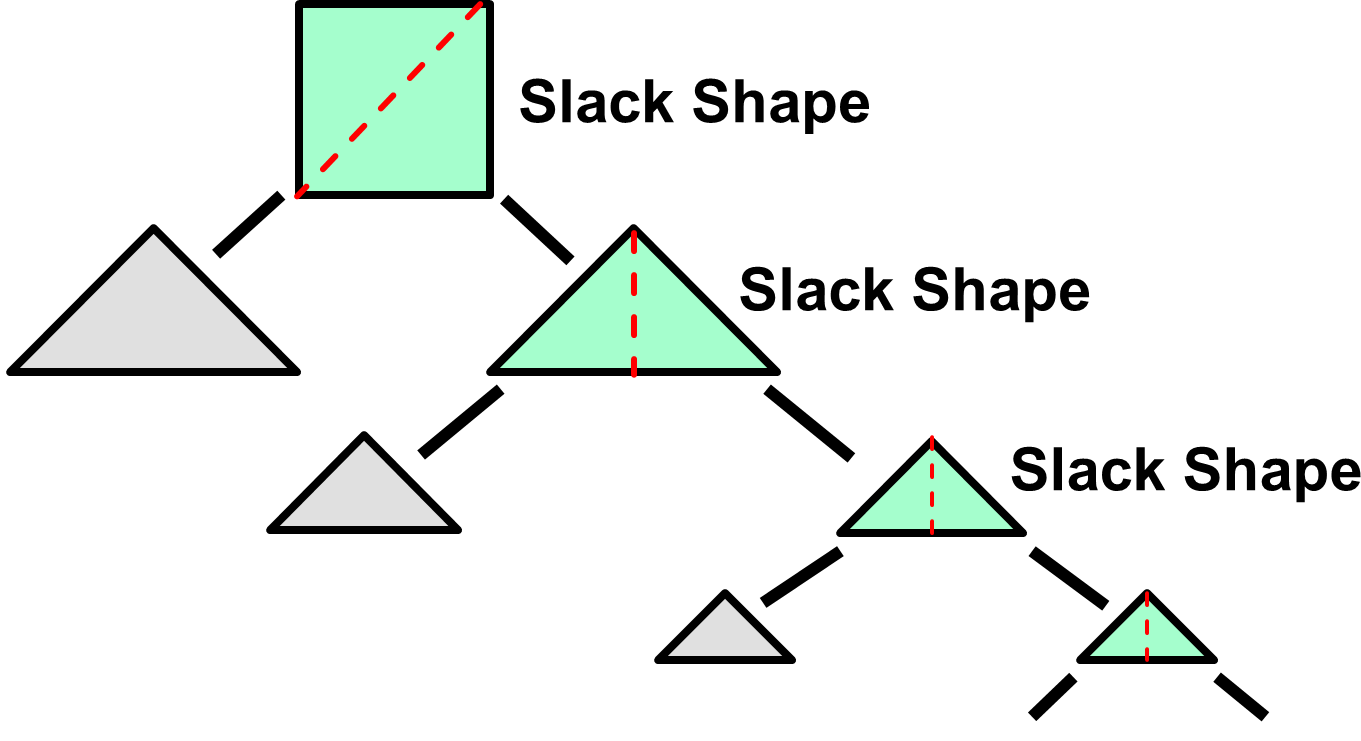}
      \caption{Initial structure containing no circles}
      \label{fig:tree_structure_initial}
    \end{subfigure}%
    \begin{subfigure}[b]{.5\linewidth}
      \centering
      \includegraphics[width=.8\linewidth]{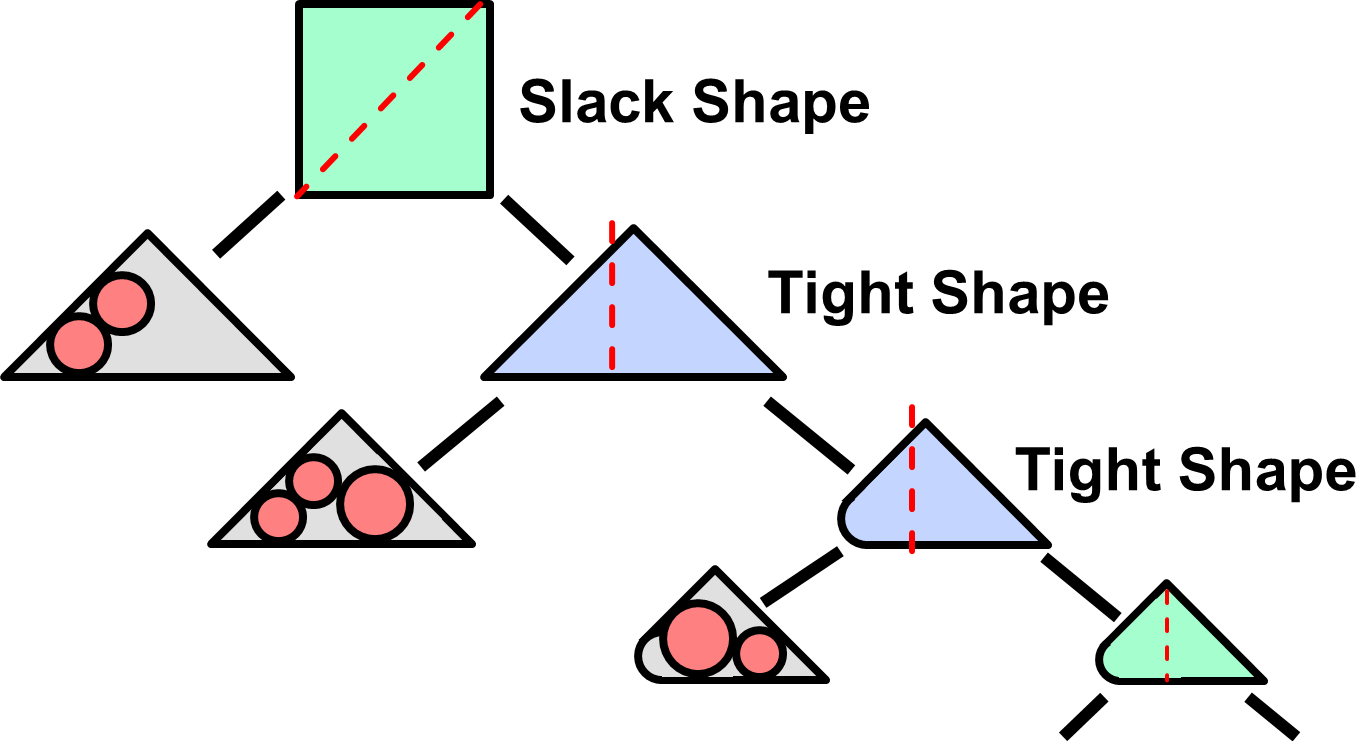}
      \caption{After circles have been inserted}
      \label{fig:tree_structure_partial}
    \end{subfigure}
  \caption{Binary tree structure used by the Online Split Packing Algorithm.}
  \label{fig:tree_structure}
\end{figure}

The root of the binary tree is the original region we are packing circles into. Each level of the binary tree has a rightmost node, which is split into two nonoverlapping subregions that become the left and right children of the node. The left child is always packed with the original Split Packing Algorithm (described later in Lemma \ref{lem:splitpackingsemihats}), while the right child, as the rightmost node of the next level, is again recursively split into two nonoverlapping subregions. The exact details of this split is explained in Section \ref{section:splittingshapes}. We always maintain this structure between insertions. An example of this can be seen in Figure \ref{fig:tree_structure_partial}. Because of this structure, we primarily focus on the rightmost node of each level. (We simply refer to them as ``rightmost nodes'' in the rest of this paper for brevity)

%Sections \ref{section:splittingtriangles} and \ref{section:splittingsquares} describe how a shape is subdivided into two nonoverlapping subregions, which serve as its two children in the binary tree. The left child of the shape will be packed with a modified version of the original Split Packing Algorithm (Described later in Lemma \ref{lem:splitpackingsemihats}), while the right child is recursively packed by our algorithm. As left children are packed with the original Split Packing Algorithm, we only concern ourselves with the rightmost node in each level of the binary tree. An example of this can be seen in Figure \ref{fig:tree_structure_partial}.

The initial configuration of the binary tree is the result of running the \Call{Repack}{} procedure defined in Section \ref{section:repackalgorithm} on the original region (the root shape) with no circles. This would build a binary tree which is an infinitely long sequence of ``ideal splits'' (Figure \ref{fig:tree_structure_initial}). Each empty square will be split along its diagonal into two right angled triangles, and each empty right angled triangle will be split by its ideal capacities (Definition \ref{dfn:idealcapacities}) into two right angled triangles as seen in Figure \ref{fig:idealcapacities}. As the tree extends indefinitely, in an actual implementation of the algorithm, the data structure for the tree only needs to be built as deep as the deepest inserted circle in the tree.

% XXXX SEE IF THIS IS NEEDED ANYWHERE. IT IS REPLACED BY THE ABOVE PARAGRAPH.
%An empty square is split along its diagonal into two right angled triangles, and an empty right angled triangle is split by its ideal capacities (Definition \ref{dfn:idealcapacities}) as seen in Figure \ref{fig:idealcapacities}. Thus, the initial state of the binary tree (representing an empty region, containing no circles) is an infinite chain of ideal splits, like shown in Figure \ref{fig:tree_structure_initial}. We note that the full height of the tree does not need to be represented in an actual implementation of the algorithm. The tree only needs to be built as deep as the deepest inserted circle, and extended when needed.

\begin{definition}[Packing Invariant]
\label{dfn:packinginvariant}
Let $S$ be a rightmost node of capacity $a$, and let $C$ be a set of circles. The packing invariant is said to hold on $S$ with circles $C$ if:
\begin{enumerate}
\item
$S$ is either a square or a $b$-semihat for some $b \geq 0$.
\item
The total size of the circles in $C$ is at most $a$.
\item
If $S$ is a $b$-semihat with $b > 0$, then there exists a circle in $C$ with size at least $b$.
\end{enumerate}
The packing invariant is said to hold on the packing if on each level of the binary tree, the packing invariant holds on its rightmost node $S$ with the circles $C$ contained in $S$.
\end{definition}

\begin{definition}[Tight and Slack Shapes]
A rightmost node $S$ of the binary tree is called a:
\begin{enumerate}
\item
\textbf{Tight Shape} if the total size of the circles contained within the left child of $S$ is equal to the capacity of the left child.
\item
\textbf{Slack Shape} if $S$ is not a tight shape and the right child of $S$ is a full triangle.
\end{enumerate}
\label{dfn:tightandslackshapes}
\end{definition}

\begin{definition}[Shape Invariant]
\label{dfn:shapeinvariant}
The Shape Invariant holds on a rightmost node $S$ if it is either a tight shape or a slack shape. The Shape Invariant holds on a binary tree if it holds on the rightmost node of every level of the tree.
%The shape invariant holds on a binary tree if on every level, the rightmost node $S$ is either a tight shape or a slack shape.
\end{definition}

Both invariants (Definition \ref{dfn:packinginvariant}, \ref{dfn:shapeinvariant}) are maintained between all insertion requests.
\longshort{The Packing Invariant is used for the proof of correctness (that the packing is valid), and will be shown to always hold in Section \ref{section:correctnessofrepack}. The Shape Invariant is used in the proof of the reallocation cost bound and algorithm termination, and will be shown to always hold in Section \ref{section:proofofcostbound}.}{The Packing Invariant is used for the proof of correctness (that the packing is valid). The Shape Invariant is used in the proof of the reallocation cost bound and algorithm termination. They will be shown to always hold in the full version of the paper.}

% XXX OLD - DELETE
% We maintain the two invariants (Definition \ref{dfn:packinginvariant}, \ref{dfn:shapeinvariant}) in between all insertion requests. Both invariants initially hold as the root node is a square, and the rightmost node on each lower level is a full triangle ($0$-semihats), so all rightmost nodes are slack shapes. There are also no circles in any of the shapes. We show in Section \ref{section:invariantsandproperties} that the shape invariant always holds, and we show in \longshort{Section \ref{section:correctnessofrepack}}{the full version of the paper} that the packing invariant always holds.

% SPLIT-----------

\subsection{Inserting a Circle}

\begin{algorithm}
\caption{Inserting a circle $c$ into a shape $S$}\label{alg:insertionalgorithm}
\begin{algorithmic}[1]
\Procedure{InsertCircle}{$c,S$}
\State $R \gets S.rightChild$
\If {$S$ is a tight shape or \Call{totalSize}{$R$} + $c.size < R.capacity$}
\State  \Call{InsertCircle}{$c,R$} \Comment{recurse on right child}
\Else 
\State  $C \gets S.containedCircles \cup \{c\}$
\State  \Call{Repack}{$C$, $S$} \Comment{repack current node, end recursion}
\EndIf 
\EndProcedure
\end{algorithmic}
\end{algorithm}

Algorithm \ref{alg:insertionalgorithm} inserts a new circle $c$ into a shape $S$. To insert a new circle $c$, we run Algorithm \ref{alg:insertionalgorithm} on the root of the binary tree. Algorithm \ref{alg:insertionalgorithm} recursively inserts the circle into the right child of each subsequent node, until a single repack is triggered at some level on its rightmost node $S$. \Call{Repack}{$C,S$} completely rebuilds a subtree rooted at $S$ to include the new circle in the packing, and incurrs a reallocation cost proportional to the total size of the circles within the subtree. Details on the \Call{Repack}{} algorithm are given in Section \ref{section:repackalgorithm}.

%, which allocates the circle to a position in the packing. \Call{Repack}{$C,S$} refers to our repacking algorithm, described in Section \ref{section:repackalgorithm}. \Call{Repack}{$C,S$}, completely rebuilds the subtree consisting of the node and its descendants, incurring a reallocation cost proportional to the total size of the circles within the subtree.

%S~>INS Refer to the full version of our paper for the proof of any stated theorems.

% \begin{algorithm}
% \caption{Inserting a circle $c$ into a shape $S$}\label{alg:insertionalgorithm_old}
% \begin{algorithmic}[1]
% \Procedure{InsertCircle}{$c,S$}
% \State $R \gets S.rightChild$
% \If {$S$ is a tight shape}  \Comment{recurse on right child}
% \State  \Call{InsertCircle}{$c,R$}
% \Else
%     \If {\Call{totalSize}{$R$} + $c.size \geq R.capacity$} \Comment{repack current node}
% \State  $C \gets S.circles \cup \{c\}$
% \State  \Call{Repack}{$C$, $S$}
%     \Else \Comment{recurse on right child}
% \State  \Call{InsertCircle}{$c,R$}
%     \EndIf
% \EndIf 
% \EndProcedure
% \end{algorithmic}
% \end{algorithm}

% \begin{algorithm}
% \caption{Inserting a circle $c$ into a shape $S$}\label{alg:insertionalgorithm}
% \begin{algorithmic}[1]
% \Procedure{InsertCircle}{$c,S$}
% \State $R \gets S.rightChild$
% \If {$S$ is not a tight shape and \Call{totalSize}{$R$} + $c.size \geq R.capacity$}
% \State  $C \gets S.circles \cup \{c\}$
% \State  \Call{Repack}{$C$, $S$} \Comment{repack current node}
% \Else 
% \State  \Call{InsertCircle}{$c,R$} \Comment{recurse on right child}
% \EndIf 
% \EndProcedure
% \end{algorithmic}
% \end{algorithm}

\begin{lemma}[Insertion Invariant]
\label{thm:insertioninvariant}
If \Call{InsertCircle}{$c,S$} (Algorithm \ref{alg:insertionalgorithm}) is called when the total size of circle $c$ and the circles already in $S$ is at most the capacity of $S$, then Algorithm \ref{alg:insertionalgorithm} will only (recursively) call \Call{InsertCircle}{$c,R$} (on Line 4) when the total size of the circle $c$ and the circles already in $R$ is at most the capacity of $R$.
\end{lemma}
\begin{proof}
Let $C_R$ be the set of circles to be packed into the right child. Let $C_L$ be the set of circles packed into the left child. Let $L$ and $R$ represent the left and right children respectively. There are two possible cases where we use the recursive call \Call{InsertCircle}{$c,R$}. If \Call{totalSize}{$R$} + $c.size < R.capacity$, it is clear the invariant is maintained. On the other hand, if $S$ is a tight shape, suppose that $\Call{totalSize}{C_R} > R.capacity$. As $S$ is a tight shape, $\Call{totalSize}{C_L} = L.capacity$, so $\Call{totalSize}{C = C_L\cup C_R} > R.capacity + L.capacity = S.capacity$, which contradicts the original assumption.
\end{proof}

The only precondition when calling \Call{InsertCircle}{$c,S$} (Algorithm \ref{alg:insertionalgorithm}) is for the total size of the new circle $c$ and the circles already in $S$ to be at most the capacity of $S$.
Lemma \ref{thm:insertioninvariant} states that the precondition will continue to be met when the algorithm recursively calls itself. This is used to show Lemma \ref{lem:insertionrepackcapacity}, which states that whenever \Call{Repack}{$C,S$} is called, the Packing Invariant holds on $S$ and its ancestors when the new circle $c$ is included in their sets of contained circles. This property is used for the algorithm's proof of correctness in \longshort{Section \ref{section:correctnessofrepack}.}{the full version of the paper}

%During insertion, the invariant described in Lemma \ref{thm:insertioninvariant} is maintained, as we only recurse if there is enough room for the circle to be inserted in the child. Lemma \ref{lem:insertionrepackcapacity} immediately follows from Lemma \ref{thm:insertioninvariant}, and is used to show that the algorithm can handle any sequence of circle insertions with total size at most the capacity of the root node.

% NOTE!!
%%%% The statement of Lemma 9 {insertionrepackcapacity} is a bit confusing.  Basically, you want to say that just *before* the repack is called, the packing invariant still holds for every node that now contains the new circle c (i.e., at or above the repack).  (Of note, the repack hasn't yet happened at the point when Lemma 9 is being invoked.)  Maybe an additional sentence would help?  Or maybe the lemma could be reworded?

\begin{lemma}
\label{lem:insertionrepackcapacity}
Suppose Algorithm \ref{alg:insertionalgorithm} is called to insert a new circle $c$ into an existing packing where the Packing Invariant (Definition \ref{dfn:packinginvariant}) holds. Assume that the total size of $c$ and the circles already in the packing is at most the capacity of the root node. Suppose Algorithm \ref{alg:insertionalgorithm} calls \Call{Repack}{$C,S$} to repack some shape $S$. Just before \Call{Repack}{$C,S$} is called, for any node $S'$ that will now contain the new circle $c$ ($S'$ can be $S$ or any ancestor of $S$), the Packing Invariant will hold on $S'$ with the circles $C'\cup\{c\}$, where $C'$ is current set of circles contained in $S'$.
%When Algorithm \ref{alg:insertionalgorithm} calls \Call{Repack}{$C,S$}, let $S'$ be any rightmost node in the tree which is either $S$ or an ancestor of $S$, and let $C'$ be the circles currently contained in $S'$. Then the Packing Invariant (Definition \ref{dfn:packinginvariant}) holds on $S'$ with the circles $C'\cup\{c\}$.
\end{lemma}

\begin{proof}
Statements 1 and 3 of the Packing Invariant will continue to hold as they previously already held on $S'$ with its original set of circles $C'$. By our assumption and applying Lemma \ref{thm:insertioninvariant} inductively through Algorithm \ref{alg:insertionalgorithm}'s recursive calls, the total size of the circles in $C'\cup\{c\}$ is at most the capacity of $S'$ when Algorithm \ref{alg:insertionalgorithm} calls \Call{Repack}{$C,S$}.
\end{proof}

%\begin{lemma}
%\label{lem:insertionrepackcapacity}
%As long as the total size of the circles in the entire packing (including the newly inserted circle) is at most the capacity of the root node, Algorithm \ref{alg:insertionalgorithm} will only call \Call{Repack}{$C,S$} when the total size of the circles in $C$ is at most the capacity of $S$.
%\end{lemma}
%
%\begin{proof}
%The insertion invariant holds at the root node by the precondition of this lemma, and thus holds through all recursive calls of the algorithm by Lemma \ref{thm:insertioninvariant}. By the insertion invariant, the total size of the circles in $C$ is at most the capacity of $S$.
%\end{proof}

% SPLIT-----------

\section{The Repack Operation}
\label{section:repackalgorithm}
The \Call{Repack}{} operation is a recursive algorithm that packs a set of circles $C$ into a single shape $S$, by completely rebuilding the subtree rooted at $S$. Calling \Call{Repack}{$C$,$S$} assumes that the Packing Invariant (Definition \ref{dfn:packinginvariant}) holds on the shape $S$ with the circles $C$. To pack these circles $C$ into $S$ ($S$ is a square or some $b$-semihat of capacity $a$), the algorithm splits the shape $S$ into the left and right children, two smaller shapes of capacities $a_L$ and $a_R$ respectively, where $a_L + a_R = a$. The set of circles $C$ is also partitioned into two sets $C_L$ and $C_R$, to be packed into the left and right children respectively. Before we can split the shape $S$, the \Call{Repack}{} algorithm, given $C$ and $S$, must first decide $a_L$, $a_R$, $C_L$ and $C_R$.

We describe how \Call{Repack}{} works for $b$-semihats and squares. A square splits into two $1$-right angled triangles, and an $s$-right angled triangle splits into two $s$-right angled triangles with the same value of $s$. In Section \ref{section:splittingtriangles} we define the ideal capacities $a_{lg}^*$ and $a_{sh}^*$ (Definition \ref{dfn:idealcapacities}), where $a_{lg}^* + a_{sh}^* = a$, which represents the best possible split of an $s$-right angled triangle.

\subsection{Repacking a Triangle}
\label{section:repackingatriangle}
\noindent
Let $S$ be a $b$-semihat (for some $b \geq 0$) of capacity $a$, formed from an $s$-right angled triangle. Suppose the ideal capacities of its long and short sides are $a_{lg}^*$ and $a_{sh}^*$ respectively. To decide $a_L$, $a_R$, $C_L$ and $C_R$, there are four cases, with Case 1 having the highest priority and Case 4 having the lowest. For Cases 3 and 4, we make use of $\delta_{sh}$ defined by the following expression:
\begin{equation}\delta_{sh} := \Big(1 - \frac{1}{2\sqrt{1+s^2}-1}\Big)^2\end{equation}
Intuitively, $\delta_{sh}$ is the largest value of $\delta$ where Case 3 can pack correctly and maintain the packing invariant. A proper explanation for $\delta_{sh}$ is given in \longshort{Lemma \ref{lem:trianglepackinglemma} and the proof of correctness of Case 3 in Lemma \ref{lem:repackcorrectnesstriangles}}{the proof of correctness of this algorithm in the full version of the paper}.

\begin{description}
\item \textbf{Case 1:} There exists a circle in $C$ of size $> a_{lg}^*$

Place the largest circle in $C_L$, and the remaining circles in $C_R$. The left child will be packed tightly ($a_L = \Call{TotalSize}{C_L}$), while the right child has the remaining space.

\item \textbf{Case 2:} $C$ has total size $\leq a_{lg}^*$

Place all the circles in $C_L$, while $C_R$ remains empty. The capacities $a_L$ and $a_R$ of the left and right children will be their respective ideal capacities, $a_{lg}^*$ and $a_{sh}^*$.

\item \textbf{Case 3:}
We first iterate through the circles $C$ from the largest to the smallest, and greedily pack the circles into the left side while keeping its total size below $a_{lg}^*$, as shown in Algorithm \ref{alg:trianglecase3packing}. This results in the total size of $C_L$ being of the form $a_{lg}^* - \delta a$, for some $\delta \geq 0$. We keep the result of Algorithm \ref{alg:trianglecase3packing} if $\delta < \delta_{sh}$. The left child will be packed tightly ($a_L = \Call{TotalSize}{C_L}$), while the right child has the remaining space. Note that this means $a_L = a_{lg}^* - \delta a$ and $a_R = a_{sh}^* + \delta a$. If $\delta \geq \delta_{sh}$, we move on to Case 4 instead.

\begin{algorithm}
\caption{TriangleCase3Packing}\label{alg:trianglecase3packing}
\begin{algorithmic}[1]
\State $C_L \gets \emptyset$
\State $C_R \gets \emptyset$
\ForEach {circle $c$ in $C$ from largest to smallest}
    \If {\Call{totalSize}{$C_L$} + $c$.size $\leq a_{lg}^*$}
\State add $c$ to $C_L$
    \Else
\State add $c$ to $C_R$
    \EndIf
\EndFor
\end{algorithmic}
\end{algorithm}

\item \textbf{Case 4:} $\delta \geq \delta_{sh}$

We place the two largest circles in $C_L$, and the remaining circles in $C_R$. The left child will be packed tightly ($a_L = \Call{TotalSize}{C_L}$), while the right child has the remaining space.
\end{description}

\noindent
Now that $a_L$ and $a_R$ are determined, the algorithm splits the semihat $S$ into two smaller semihats with capacities $a_L$ and $a_R$. Section \ref{section:splittingtriangles} details how the semihat is divided up, to create two children termed the long child and the short child. The long child, of capacity $a_L$, will be the left child, and the short child, of capacity $a_R$, will be the right child. The children are then packed as follows:
\begin{enumerate}
\item The circles $C_L$ are packed into the left child using the Split Packing Algorithm for Semihats in Lemma \ref{lem:splitpackingsemihats}.
\item The circles $C_R$ are recursively packed into the right child with the \Call{Repack}{} algorithm.
\end{enumerate}

\noindent
The following \longshort{lemma is}{lemmas are} used for the \Call{Repack}{} Algorithm on triangles. Lemma \ref{lem:splitpackingsemihats} describes how the original Split Packing Algorithm can be modified to pack circles into a semihat. (The Split Packing Algorithm works for $b$-hats, but is not defined on semihats)

\begin{lemma}[Split Packing Algorithm for Semihats]
Suppose that the Split Packing Algorithm in \cite{splitpackingsoda, splitpacking} is able to pack a set of circles $C$ into a $b'$-hat with capacity $a$. If at least one circle in $C$ of size at least $b$, then the circles $C$ can be packed into a $b$-$b'$-semihat of capacity $a$.
\label{lem:splitpackingsemihats}
\end{lemma}

%S~>INS \noindent \textbf{Proof Sketch:} A $b$-$b'$-semihat is a $b'$ hat where the $b'$-curve of the long side has been enlarged to form a $b$-curve. As mentioned in Section \ref{section:background}, we can choose which child \Call{Split}{} packs the largest circle into at each level. We apply the Split Packing algorithm to the underlying $b'$ hat, but make \Call{Split}{} always pack the largest circle $c$ into the long child. This will create a valid packing for the $b$-$b'$-semihat.

%S~CROP_START - Diagram only used for proof
\begin{figure}[!h]
  \centering
  \includegraphics[width=.35\linewidth]{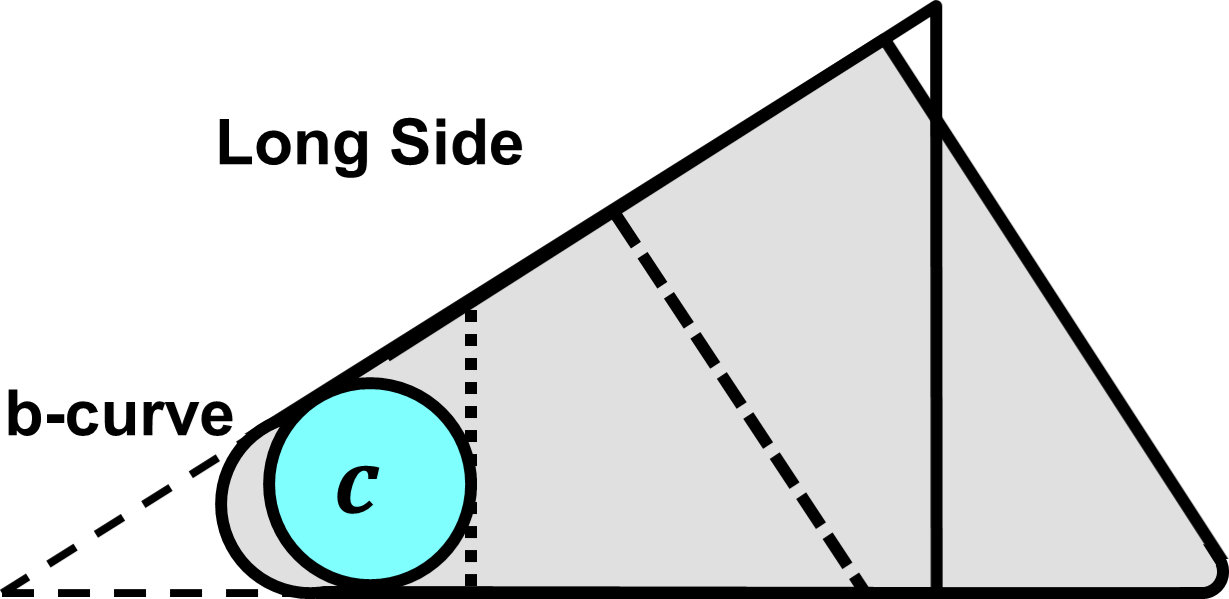}
  \caption{A $b$-$b'$-semihat. The dotted lines show the long child from the split at each level (the $b$-curve always ends up in the long child). The short children are not drawn.}
  \label{fig:splitpackingforsemihats}
\end{figure}
%S~CROP_END - Diagram only used for proof

\begin{proof}
A $b$-$b'$-semihat is a $b'$ hat where the $b'$-curve of the long side has been enlarged to form a $b$-curve. We treat the $b$-$b'$-semihat as a $b'$ hat and use the Split Packing Algorithm to pack circles into it. As described in Section \ref{section:background}, the Split Packing Algorithm recursively splits the hat into smaller hats, and terminates with one circle in each leaf node of the resulting binary tree. At each level, we find that the long child will always be the hat containing the $b$-curve (Figure \ref{fig:splitpackingforsemihats}. Due to how the \Call{Split}{} algorithm (Algorithm \ref{alg:split}) works, we can force the Split Packing Algorithm to always assign the largest circle $c$ to the long child. This way, we end up with a packing of the original $b'$ hat, where the largest circle has been packed snugly into the $b$-curve of the original $b$-$b'$-semihat. Thus all circles in $C$ will also be within the bounds of the $b$-$b'$-semihat.
%A $b$-$b'$-semihat is a $b'$ hat where the $b'$-curve of the long side has been enlarged to form a $b$-curve. At each level of the recursion, the Split Packing Algorithm splits a hat into two smaller hats (which we can refer to as the long and short child, depending on which side it comes from), and greedily assigns circles to be packed recursively into the smaller hats (using the \Call{Split}{} procedure, detailed in the Split Packing papers). At each level, the long child will be the hat containing the $b$-curve of the original hat. As the largest circle is always the first circle to be assigned, and the assignment of the first circle is arbitrary, we can always choose to assign the largest circle to the long side, which will contain the $b$-curve. This way, we will end up with a packing of the original $b'$ hat, where the largest circle has been packed snugly into the $b$-curve of the original $b$-$b'$-semihat. Thus all circles in $C$ will also be within the bounds of the $b$-$b'$-semihat.
\end{proof}

\subsection{Repacking a Square}
\label{section:repackingasquare}
\noindent
Let $S$ be a square of capacity $a$. To decide $a_L$, $a_R$, $C_L$ and $C_R$, there are two cases.
\begin{description}
\item \textbf{Case 1:} $C$ has total size $> 0.5a$

We start with all the circles in the left side. We iterate through the circles $C$ from the largest to the smallest, and greedily remove circles from the left side while keeping its total size above $0.5a$, as shown in Algorithm \ref{alg:squarecase1packing}. This results in the total size of $C_L$ being of the form $a(0.5 + \delta)$, for some $\delta \geq 0$.
\begin{algorithm}
\caption{SquareCase1Packing}\label{alg:squarecase1packing}
\begin{algorithmic}[1]
\State $C_L \gets C$
\State $C_R \gets \emptyset$
\ForEach {circle $c$ in $C$ from largest to smallest}
    \If {\Call{totalSize}{$C_L$} - $c$.size $\geq 0.5a$}
\State remove $c$ from $C_L$
\State add $c$ to $C_R$
    \EndIf
\EndFor
\end{algorithmic}
\end{algorithm}

The left child will be packed tightly ($a_L = \Call{TotalSize}{C_L}$), while the right child has the remaining space. Note that this means $a_L = a(0.5 + \delta)$ and $a_R = a(0.5 - \delta)$.

\item \textbf{Case 2:} $C$ has total size $\leq 0.5a$

Place all the circles in $C_L$, while $C_R$ remains empty. The capacities $a_L$ and $a_R$ of the left and right children will each be $0.5a$.
\end{description}

\noindent
Section \ref{section:splittingsquares} describes how a square is split into two triangles given $a_L$ and $a_R$. The left and right children have capacities $a_L$ and $a_R$ respectively. They are then packed as follows:
\begin{enumerate}
\item The circles $C_L$ are packed into the left child with the Split Packing Algorithm.
\item The circles $C_R$ are recursively packed into the right child with the \Call{Repack}{} algorithm.
\end{enumerate}

% SPLIT-----------

\subsection{Splitting Shapes}
\label{section:splittingshapes}

Splitting shapes refers to subdividing a shape into two nonoverlapping smaller shapes, given target capacities $a_L$ and $a_R$. These smaller shapes will be of capacities $a_L$ and $a_R$ respectively and stay within the bounds of the original shape. These smaller shapes, as the children of the original shape in the binary tree, can then be recursed into in the \Call{Repack}{} algorithm.

\subsubsection{Splitting Triangles}
\label{section:splittingtriangles}

The Split Packing papers \cite{splitpackingsoda, splitpacking} describe how a triangle (hat) splits into two smaller triangles (hats). Details are given in Section \ref{section:background}. We do the same thing with semihats (Figures \ref{fig:semihatsplitleft}, \ref{fig:semihatsplitright}). On an $s$-right angled triangle, we call the two smaller semihats the long and short children, corresponding to whether they came from the long or short side.

%S~>INS  \begin{figure}[!h]
%S~>INS    \centering
%S~>INS      \begin{subfigure}[b]{.33\linewidth}
%S~>INS        \centering
%S~>INS        \includegraphics[width=.96\linewidth]{diagrams/semihatsplitleft_v2.png}
%S~>INS        \caption{$a_{lg} \geq a_{lg}^*$ ($a_{sh} \leq a_{sh}^*$)}
%S~>INS        \label{fig:semihatsplitleft}
%S~>INS      \end{subfigure}%
%S~>INS      \begin{subfigure}[b]{.33\linewidth}
%S~>INS        \centering
%S~>INS        \includegraphics[width=.96\linewidth]{diagrams/semihatsplitright_v2.png}
%S~>INS        \caption{$a_{lg} < a_{lg}^*$ ($a_{sh} > a_{sh}^*$)}
%S~>INS        \label{fig:semihatsplitright}
%S~>INS      \end{subfigure}%
%S~>INS      \begin{subfigure}[b]{.33\linewidth}
%S~>INS        \centering
%S~>INS        \includegraphics[width=.96\linewidth]{diagrams/squaresplit_v2.png}
%S~>INS        \caption{Splitting a square}
%S~>INS        \label{fig:squaresplit}
%S~>INS      \end{subfigure}
%S~>INS    \caption{\textbf{(a),(b):} The two ways a semihat splits into two semihats of capacities $a_{lg}$ and $a_{sh}$ respectively. \textbf{(c):} A square splits into a $b$-hat and a full triangle of capacities $a_L$ and $a_R$ respectively.}
%S~>INS  \end{figure}

\begin{lemma}
\label{lem:idealcapacitysizes}
The ideal capacities (Definition \ref{dfn:idealcapacities}) $a_{lg}^*$ and $a_{sh}^*$ of an $s$-right angled triangle of capacity $a$ obey the following properties:\noindent
\begin{multicols}{3}
\begin{enumerate}
\item $\ds a_{lg}^* = a\frac{s^2}{1+s^2}$ 
\item $\ds a_{sh}^* = a\frac{1}{1+s^2}$ 
\item $\ds a_{lg}^* + a_{sh}^* = a$ 
\end{enumerate}
\end{multicols}
\end{lemma}
\begin{proof}
Referring to Figure \ref{fig:idealcapacities}, as the three triangles are similar, the shortest side of the long child will be of length $\ds \frac{s\ell}{\sqrt{1+s^2}}$, so $\ds \frac{a_{lg}^*}{a} = \big(\frac{s}{\sqrt{1+s^2}}\big)^2 = \frac{s^2}{1+s^2}$. Similarly, the shortest side of the short child will be of length $\ds \frac{\ell}{\sqrt{1+s^2}}$, so $\ds \frac{a_{sh}^*}{a} = \big(\frac{1}{\sqrt{1+s^2}}\big)^2 = \frac{1}{1+s^2}$.

Finally, $\ds a_{lg}^* + a_{sh}^* = a\frac{s^2}{1+s^2} + a\frac{1}{1+s^2} = a$.
\end{proof}

Suppose a semihat with capacity $a$ is split into long and short children with capacities $a_{lg}$ and $a_{sh}$ respectively. The Split Packing papers \cite{splitpackingsoda, splitpacking} show that if $a_{lg} + a_{sh} \leq a$, the two children will not overlap. In our algorithm, we always have $a_{lg} + a_{sh} = a$. If $a_{lg} \geq a_{lg}^*$, a $b$-semihat splits into a $b$-$b_{lg}$-semihat on the long side (for some $b_{lg} \geq 0$), and a full triangle on the short side (Figure \ref{fig:semihatsplitleft}). If $a_{lg} \leq a_{lg}^*$, a $b$-semihat splits into a $b$-semihat on the long side, and a $b_{sh}$-semihat on the short side (for some $b_{sh} \geq 0$, Figure \ref{fig:semihatsplitright}). Structurally, this ensures the first property of the Packing Invariant (Definition \ref{dfn:packinginvariant}), as the right child is always a $b$-semihat for some value of $b$ (a full triangle is a $0$-semihat).

%S~CROP_START - Separate Splitting Diagrams
\begin{figure}[!h]
  \centering
    \begin{subfigure}[b]{.5\linewidth}
      \centering
      \includegraphics[width=.7\linewidth]{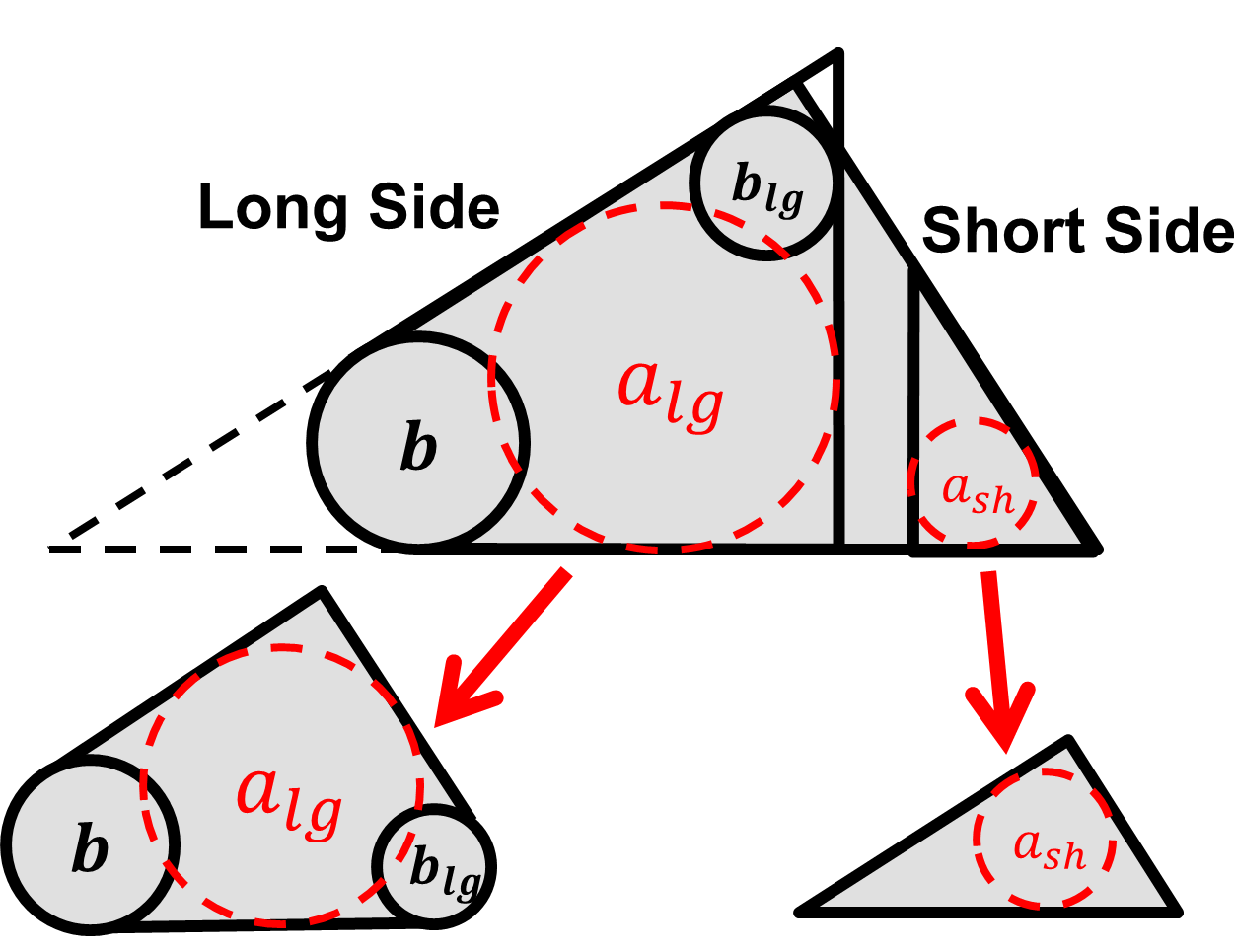}
      \caption{$a_{lg} \geq a_{lg}^*$ (and consequently $a_{sh} \leq a_{sh}^*$)}
      \label{fig:semihatsplitleft}
    \end{subfigure}%
    \begin{subfigure}[b]{.5\linewidth}
      \centering
      \includegraphics[width=.7\linewidth]{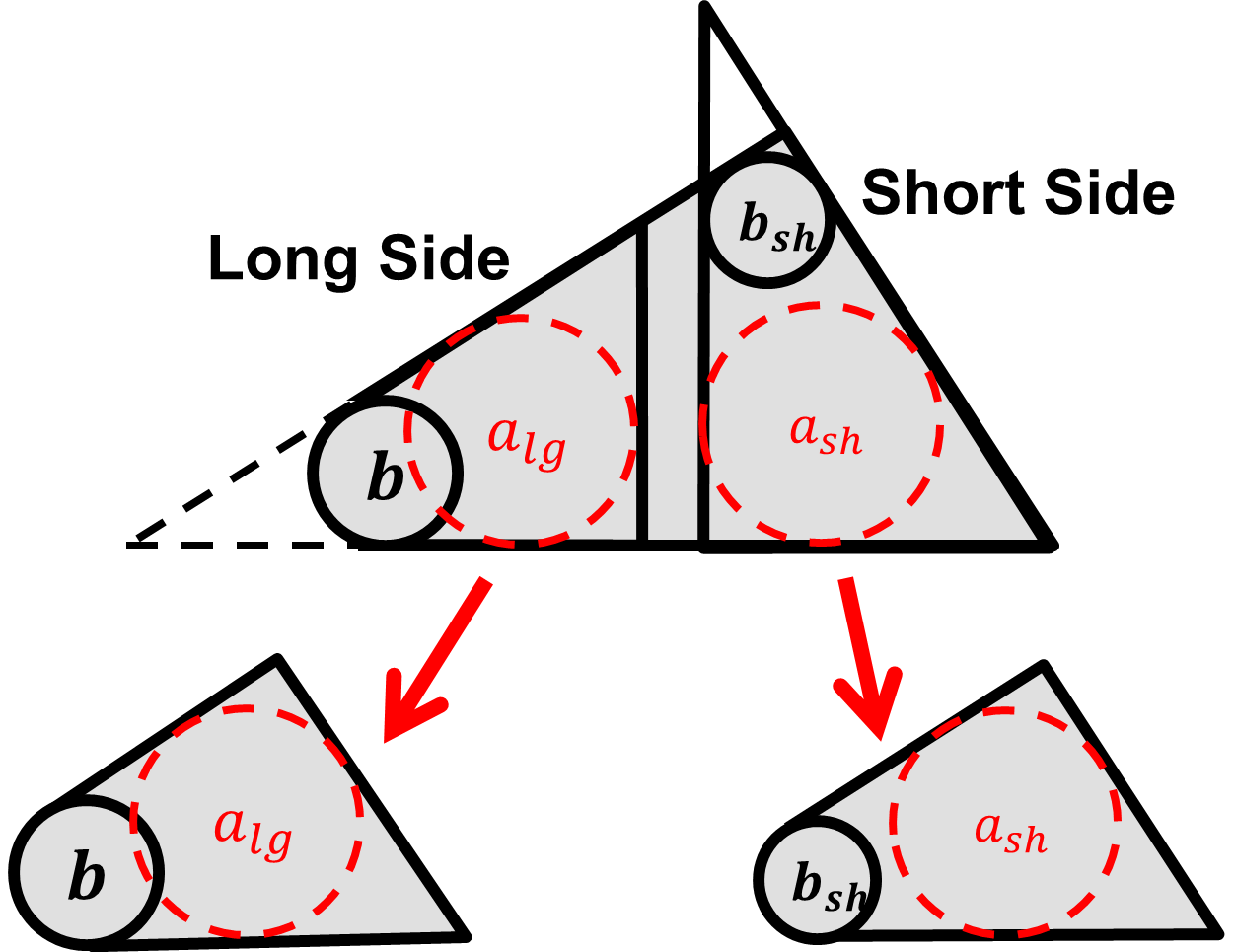}
      \caption{$a_{lg} < a_{lg}^*$ (and consequently $a_{sh} > a_{sh}^*$)}
      \label{fig:semihatsplitright}
    \end{subfigure}
  \caption{The two ways a semihat splits into two semihats of capacities $a_{lg}$ and $a_{sh}$ respectively.}
  \label{fig:semihatsplit}
\end{figure}
%S~CROP_END - Separate Splitting Diagrams

\subsubsection{Splitting Squares}
\label{section:splittingsquares}
Our strategy for splitting squares is also similar to Split Packing, which is explained in Section \ref{section:background}. A square splits into two isosceles right angled triangles ($1$-right angled triangles), as shown in Figure \ref{fig:squaresplit}.

%S~CROP_START - Separate Splitting Diagrams
\begin{figure}[!h]
  \centering
  \includegraphics[width=.35\linewidth]{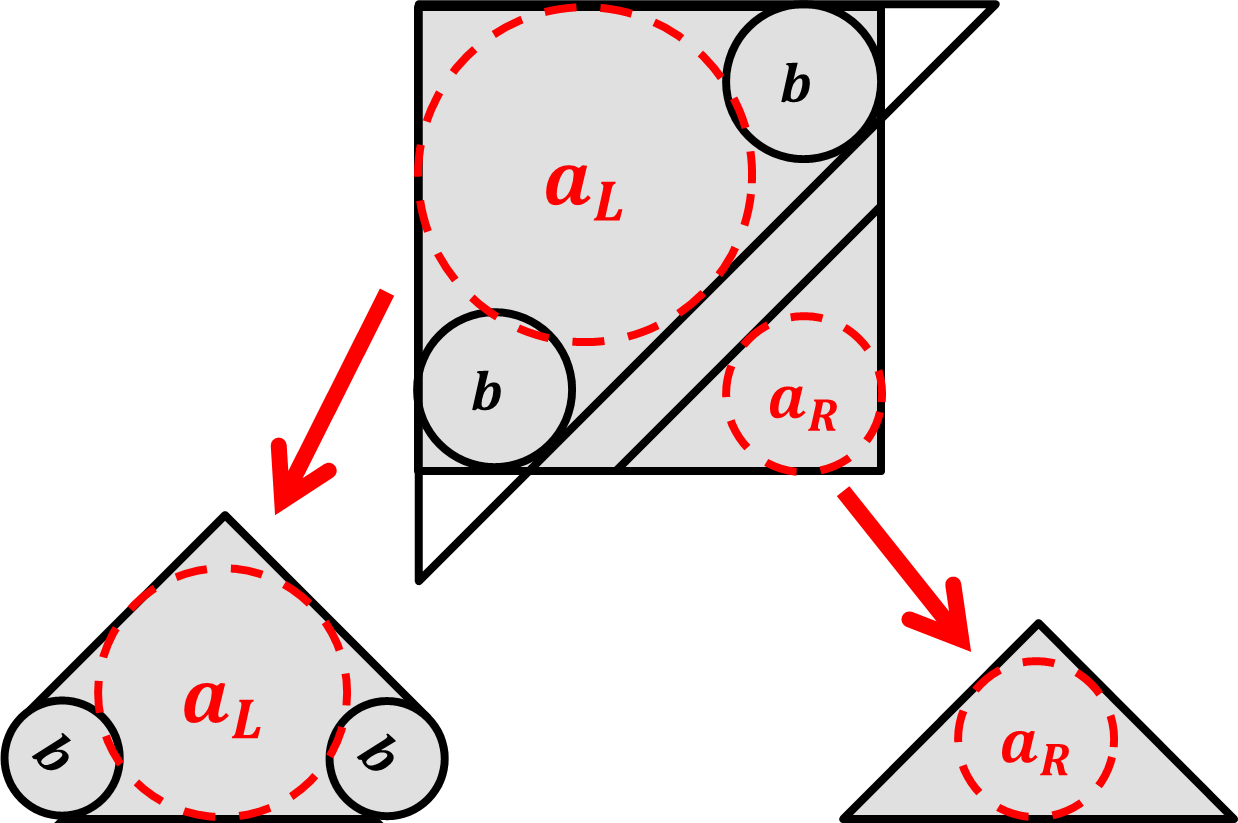}
  \caption{Splitting a square into a $b$-hat and a full triangle of capacities $a_L$ and $a_R$ respectively.}
  \label{fig:squaresplit}
\end{figure}
%S~CROP_END - Separate Splitting Diagrams

% SPLIT-----------

%S~CROP_START - Proof of Correctness
\section{Proof of Correctness}
To prove the correctness of our algorithm (for inserts only), assuming that the total size of all inserted circles is at most the capacity of the root node, we need to show that:
\begin{enumerate}
\item Every circle will be assigned a position in the packing.
\item All circles are packed within the bounds of the region.
\item No two circles overlap in the packing.
\end{enumerate}

Statement 1 will be proved in Section \ref{section:termination}. The latter two will be proved in Section \ref{section:correctnessofrepack}.

\subsection{Termination}
\label{section:termination}

Algorithm \ref{alg:insertionalgorithm} is a recursive algorithm that only terminates upon calling \Call{Repack}{}. The \Call{Repack}{} is a process which never terminates, instead recursively building an infinite sequence of empty nodes after the last circle has been allocated. To prove that the insertion algorithm will always allocate the new circle (as well as the displaced existing circles) as position in the packing (contained in some node), we show that, whenever we want to insert a new circle,
\begin{enumerate}
\item Algorithm \ref{alg:insertionalgorithm} will always terminate (by calling \Call{Repack}{}) after a finite number of steps.
\item Every circle that is to be repacked when \Call{Repack}{} is called will be assigned a position after a finite number of steps.
\end{enumerate}

These two statements will be proven as Corollary \ref{prop:inserttermination} and Theorem \ref{prop:repacktermination}. Before we prove these properties, we first show that the Shape Invariant, along with some other properties, always hold on the binary tree (Lemma \ref{lem:shapeinvariantholds}).

\begin{lemma}
\label{lem:shapeinvariantholds}
The Shape Invariant (Definition \ref{dfn:shapeinvariant}) will hold on the initial configuration of the binary tree, as well as after each insertion.
\noindent
Furthermore, the following properties will also always hold on all rightmost nodes:
\begin{enumerate}
\item If a node $S$ is a slack shape, the capacity of the right child is $\frac{1}{1+s^2}$ of the capacity of $S$.
\item If a node $S$ is a tight shape, there will be at least one circle in its left child.
\end{enumerate}
\end{lemma}

\begin{proof}
The \Call{Repack}{} algorithm can only produce tight or slack shapes. Only Case 2 of \Call{Repack}{} on both triangles and squares can generate slack shapes, as the right child will be a full triangle. All other cases produce tight shapes. Let $a$ be the capacity of the original shape and $a_R$ be the capacity of the right child. For Case 2 on triangles, we have $a_R = a_{sh}^* = a\frac{1}{1+s^2}$ (Lemma \ref{lem:idealcapacitysizes}). For Case 2 on squares, we have $a_R = 0.5a = a\frac{1}{1+s^2}$. Thus for all cases which produce slack shapes, the capacity of the slack shape's right child is $\frac{1}{1+s^2}$ of the capacity of the original shape. For all cases which produce tight shapes, at least one circle will be placed in the shape's left child.

As stated in Section \ref{section:onlinesplitpackingalgorithm}, the initial configuration of the binary tree is the result of runnning the \Call{Repack}{} algorithm on the root shape with no circles. Thus the Shape Invariant as well as both properties would hold on the initial configuration of the binary tree.

The binary tree can only be modified through calling the \Call{Repack}{} algorithm on some shape $S$, which rebuilds the subtree from $S$ and below. Thus the Shape Invariant, as well as both properties, will hold on $S$ and its descendant rightmost nodes. For each ancestor of $S$, the left child will not be affected by a repack of $S$, and the shape of the right child (which might be $S$) will not be changed by a repack of $S$. The Shape Invariant as well as both properties, which held previously, will continue to hold on them.
\end{proof}

By Lemma \ref{lem:shapeinvariantholds}, the capacity of the right child of a slack $s$-shape is $\frac{1}{1+s^2}$ of its parent's capacity. This allows us to prove Lemma \ref{lem:maxslackshapes} and Corollary \ref{prop:inserttermination}.

\begin{lemma}
When inserting a circle of size $c$ into an $s$-shape of capacity $a$, Algorithm \ref{alg:insertionalgorithm} passes through at most $\floor{\log_{1+s^2} \frac{a}{c}} + 1$ slack shapes before terminating by calling \Call{Repack}{}.
\label{lem:maxslackshapes}
\end{lemma}

\begin{proof}
If the current shape is a slack shape, the ratio of the capacity of the right child to its parent's is $\frac{1}{1+s^2}$. As each slack shape reduces the shape size by a factor of $1+s^2$, the circle passes through at most $\floor{\log_{1+s^2} \frac{a}{c}} + 1$  slack shapes before the circle itself must be larger than the right child's capacity.
\end{proof}

\begin{corollary}
Algorithm \ref{alg:insertionalgorithm} always terminates after a finite number of steps.
\label{prop:inserttermination}
\end{corollary}

\begin{proof}
Algorithm \ref{alg:insertionalgorithm} terminates when it calls \Call{Repack}{}. By Lemma \ref{lem:shapeinvariantholds}, each tight shape has at least one circle in its left child. Thus, there at most $n$ tight shapes, where $n$ is the amount of circles inserted so far. By Lemma \ref{lem:maxslackshapes}, a circle passes through a finite number of slack shapes. As every rightmost node is either tight or slack, the algorithm will terminate and call \Call{Repack}{} after a finite number of steps.
\end{proof}

\begin{theorem}
Every circle that is to be repacked when \Call{Repack}{} is called will be assigned a position after a finite number of steps.
\label{prop:repacktermination}
\end{theorem}

\begin{proof}
The \Call{Repack}{} procedure assigns a circle a position when it allocates the circle to the left child instead of the right. This is as the left child is packed with the original Split Packing Algorithm. Every time the \Call{Repack}{} procedure creates a tight shape, at least one circle will be assigned to the left child. This limits the number of tight shapes created to the number of circles to be packed. Therefore, \Call{Repack}{} will create a slack shape after a finite number of recursions, and in all cases where a slack shape is created, all remaining circles will be packed into the left child. Therefore, all circles will be assigned a position after a finite number of steps.
\end{proof}

\subsection{Properties of Triangle and Square Splitting}
In this section, we show some important relationships relating to the sizes of the $b$-curves generated when our triangles (including semihats) and squares are split in the manner described in Section \ref{section:splittingshapes}. The Lemmas in this section are used to prove the validity of our packing in the following sections.

\subsubsection{Triangles}
\begin{figure}[!h]
  \centering
    \begin{subfigure}[b]{.36\linewidth}
      \centering
      \includegraphics[width=.9\linewidth]{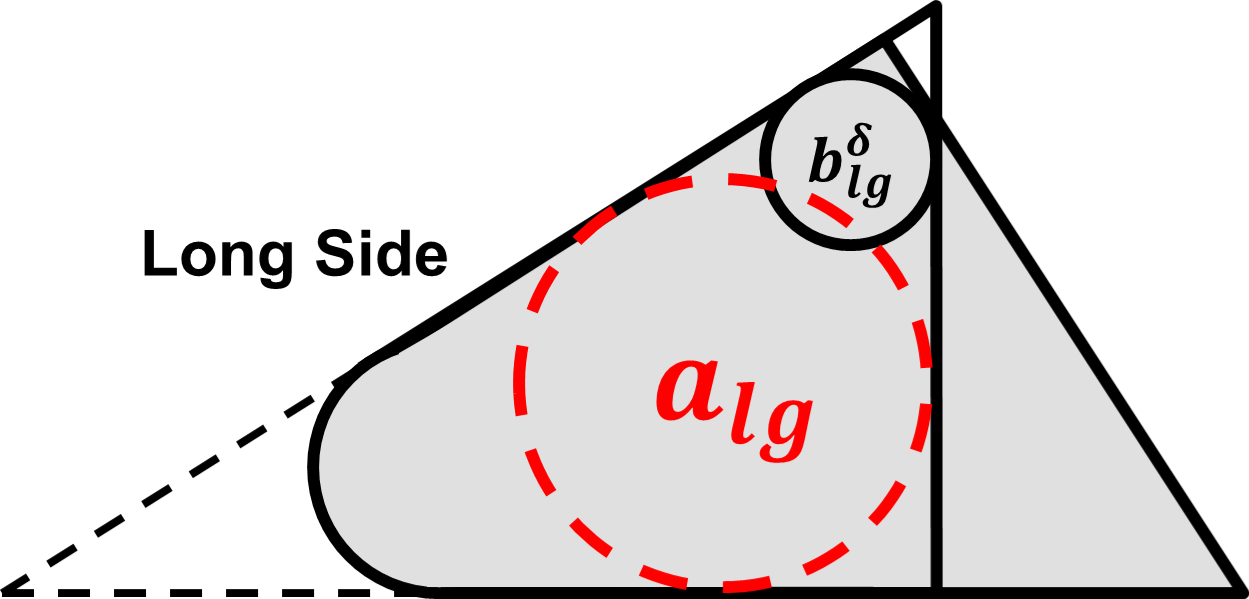}
      \caption{Long side: $a_{lg} = a_{lg}^* + \delta a$}
      \label{fig:splittingproperties_long}
    \end{subfigure}%
    \begin{subfigure}[b]{.36\linewidth}
      \centering
      \includegraphics[width=.9\linewidth]{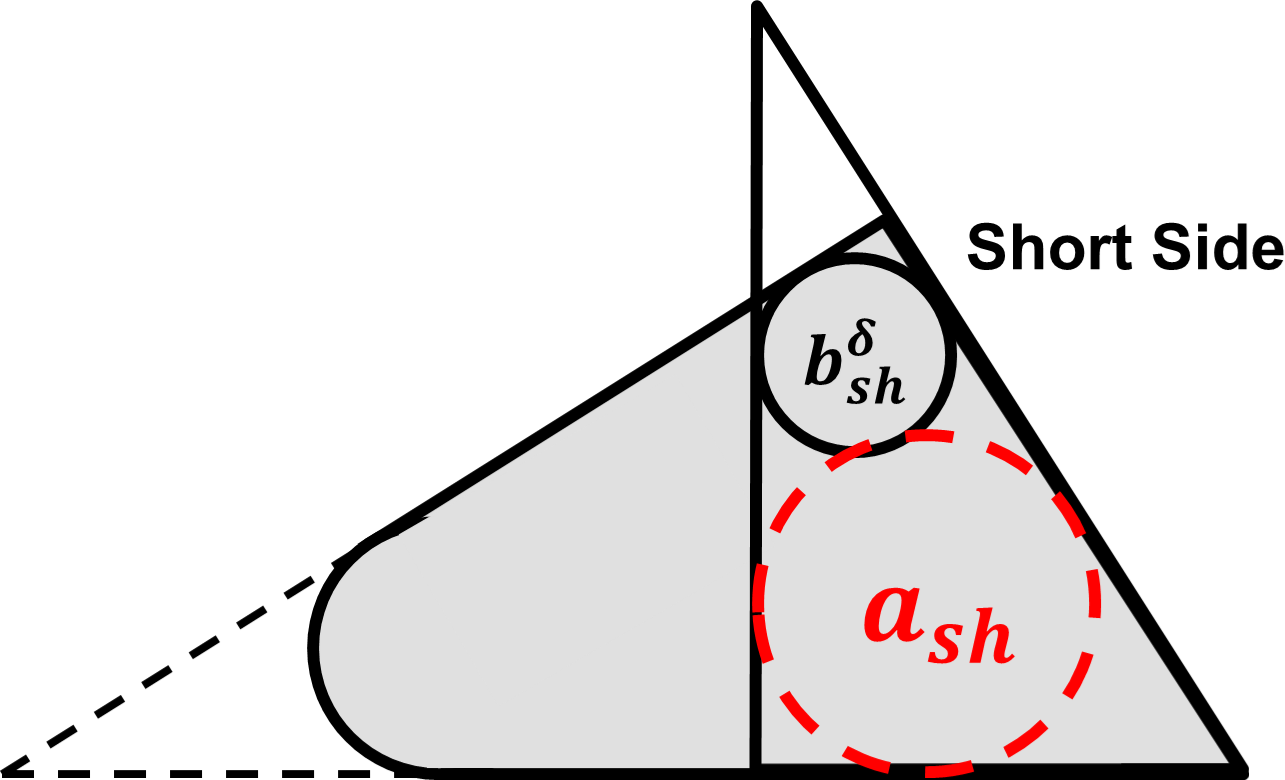}
      \caption{Short side: $a_{sh} = a_{sh}^* + \delta a$}
      \label{fig:splittingproperties_short}
    \end{subfigure}%
    \begin{subfigure}[b]{.28\linewidth}
      \centering
      \includegraphics[width=.8\linewidth]{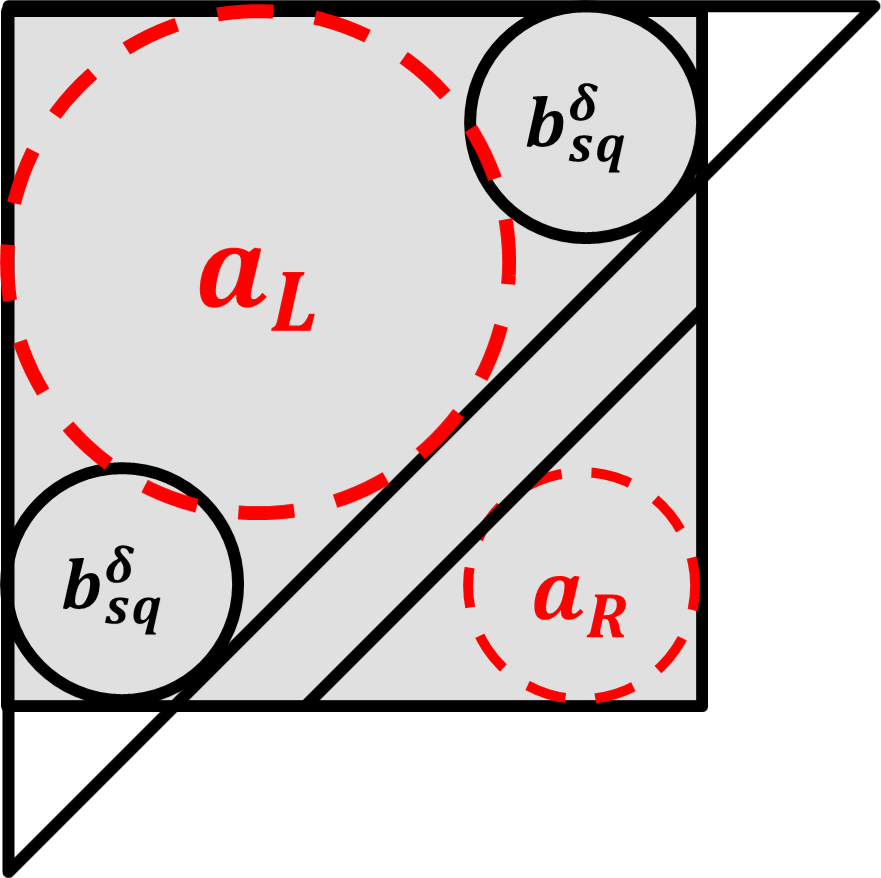}
      \caption{Square: $a_L = 0.5a + \delta a$}
      \label{fig:splittingproperties_square}
    \end{subfigure}
  \caption{Necessary $b$-curves needed for child nodes to fit within their parent shapes.}
  \label{fig:tree_structure}
\end{figure}

\begin{lemma}[Relationship between capacity of a child and its $b$-curve (triangles)]
Consider an $s$-right angled triangle $S$ of capacity $a$ split into long and short children of capacities $a_{lg}$ and $a_{sh}$ respectively. Let $a_{lg}^*$ and $a_{sh}^*$ be the ideal capacities of the triangle (Definition \ref{dfn:idealcapacities}). If $a_{lg} \geq a_{lg}^*$, we write $a_{lg} = a_{lg}^* + \delta a$ for some $\delta \geq 0$. Let $b_{lg}^\delta$ be the smallest $b$ such that a long child of capacity $a_{lg}$ with a $b$-curve on the short side can fit within the bounds of the triangle $S$ (Figure \ref{fig:splittingproperties_long}). Then $b_{lg}^\delta$ has the following expression in terms of $\delta$:
\begin{equation} \displaystyle b_{lg}^\delta = a\Big(\frac{\sqrt{s^2+\delta(1+s^2)} - s}{\sqrt{1+s^2} - s}\Big)^2 \end{equation}
Similarly, if we instead have $a_{sh} \geq a_{sh}^*$, we can write $a_{sh} = a_{sh}^* + \delta a$ for some $\delta \geq 0$. Let $b_{sh}^\delta$ be the smallest $b$ such that a short child of capacity $a_{sh}$ with a $b$-curve on the long side can fit within the bounds of the triangle $S$ (Figure \ref{fig:splittingproperties_short}). Then $b_{sh}^\delta$ has the following expression in terms of $\delta$:
\begin{equation} \displaystyle b_{sh}^\delta = a\Big(\frac{\sqrt{1+\delta(1+s^2)} - 1}{\sqrt{1+s^2} - 1}\Big)^2 \end{equation}
\label{lem:trianglebsize}
\end{lemma}
\begin{proof}
Appendix \ref{proof:trianglebsize}
\end{proof}

\begin{lemma}[Triangle Packing Lemma (Short Child)]
Let $b_{sh}^\delta$ be the expression defined in terms of $\delta \geq 0$ (short child) in Lemma \ref{lem:trianglebsize} (where $s \geq 1$, $a > 0$).  Define $\delta_{sh}$ as follows:
\begin{equation} \displaystyle \delta_{sh} := \Big(1 - \frac{1}{2\sqrt{1+s^2}-1}\Big)^2 \end{equation}
Then for all $\delta \in [0,\delta_{sh}]$, we have $\delta \geq b_{sh}^\delta/a$.
\label{lem:trianglepackinglemma}
\end{lemma}
\begin{proof}
Appendix \ref{proof:trianglepackinglemma}
\end{proof}

%\subsection{Square Splitting - Properties}

\subsubsection{Squares}
\begin{lemma}[Relationship between capacity of a child and its $b$-curve (squares)]
Consider a square $S$ of capacity $a$ with a child of capacity $a_L$. If $a_L \geq 0.5a$, we write $a_L = 0.5a + \delta a$ for some $\delta \geq 0$. Let $b_{sq}^\delta$ be the smallest $b$ such that a child of capacity $a_L$ and a $b$-curve on both sides can fit within the bounds of the square $S$ (Figure \ref{fig:splittingproperties_square}). Then $b_{sq}^\delta$ has the following expression in terms of $\delta$:
\begin{equation} \displaystyle b_{sq}^\delta = a\Big(\frac{\sqrt{1+2\delta} - 1}{\sqrt{2} - 2} \Big)^2 \end{equation}
\label{lem:squarebsize}
\end{lemma}
\begin{proof}
Appendix \ref{proof:squarebsize}
\end{proof}

\begin{lemma}[Square Packing Lemma]
Let $b_{sq}^\delta$ be the expression defined in terms of $\delta \geq 0$ in Lemma \ref{lem:squarebsize} (where $a > 0$). We then have $\delta \geq b_{sq}^\delta/a$ for $\delta \in [0,0.5]$.
\label{lem:squarepackinglemma}
\end{lemma}
\begin{proof}
Appendix \ref{proof:squarepackinglemma}
\end{proof}

% SPLIT-----------

\subsection{Correctness of Repack}
\label{section:correctnessofrepack}

We prove Theorem \ref{thm:correctness} to hold after every step of the packing. Theorem \ref{thm:correctness} necessarily implies that the root node of the tree is validly packed. As we have already shown that every circle will be allocated to some node, this would conclude our proof of correctness.

\begin{definition}[Valid Packing within Node $S$]
\label{dfn:validpacking}
A node $S$ is said to be validly packed if:
\begin{enumerate}
\item All circles contained in $S$ are placed within the bounds of the node.
\item No two circles contained in $S$ overlap.
\end{enumerate}
\end{definition}

\begin{theorem}
\label{thm:correctness}
Assume that the total size of all inserted circles is within the root node's capacity. Then on the binary tree's initial configuration, and after every subsequent insertion,
\begin{enumerate}
\item Every node in the tree is validly packed (Definition \ref{dfn:validpacking}).
\item The Packing Invariant (Definition \ref{dfn:packinginvariant}) holds on the packing.
\end{enumerate}
\end{theorem}

The key idea behind the proof for Theorem \ref{thm:correctness} is that our binary tree is only modified through calls to \Call{Repack}{$C$,$S$}, and only when the Packing Invariant (Definition \ref{dfn:packinginvariant}) holds the shape $S$ to be repacked on the circles $C$ to be packed into it. Before we prove Theorem \ref{thm:correctness}, we first prove some properties of a binary tree that is only modified through \Call{Repack}{} calls, in the form of Lemmas \ref{lem:childrenwithinbounds}, \ref{lem:repackcorrectnesstriangles} and \ref{lem:repackcorrectnesssquares}. We prove Theorem \ref{thm:correctness} at the end of this section.

\begin{lemma}
\label{lem:childrenwithinbounds}
If \Call{Repack}{$C$,$S$} is only called when the Packing Invariant holds on $S$ with the circles $C$, then for each node in the binary tree, the children of the node are non-overlapping and within the bounds of the node.
\end{lemma}

\begin{proof}
The binary tree is only modified through calls to \Call{Repack}{}. Calling \Call{Repack}{$C$,$S$} does not change the shape of $S$, and only rebuilds $S$'s descendants. From the description of the triangle/square splitting procedures in Section \ref{section:splittingshapes}, the only requirement for the shapes to be nonoverlapping and within the bounds of the original shape is for $a_L + a_R \leq a$ and for $a_L$ and $a_R$ to both be nonnegative and no larger than $a$.
These two requirements can be easily seen to be met in each of the cases given in Sections \ref{section:repackingatriangle} and \ref{section:repackingasquare}. (in fact, in most of the cases, we define $a_R$ to be $a - a_L$.)
\end{proof}

\begin{lemma} [Properties of \Call{Repack}{} on Triangles ($b$-semihats)]
Assume that the Packing Invariant (Definition \ref{dfn:packinginvariant}) holds on a $b$-semihat $S$ with a set of circles $C$. Suppose that we call \Call{Repack}{$C$,$S$}. Let $C_R \subseteq C$ be the set of circles that \Call{Repack}{$C$,$S$} will pack into the right child of $S$. Then,
\begin{enumerate}
\item All circles in the left child of $S$ will be packed within the bounds of the left child, and no two of these circles will overlap.
\item The packing invariant holds on the right child of $S$ with the circles $C_R$ (i.e. it holds when \Call{Repack}{} recurses into the right child)

%\item The algorithm packs all the circles in the left child within its bounds.
%\item The packing invariant (Definition \ref{dfn:packinginvariant}) is maintained when recursing into the right child.
\end{enumerate}
\label{lem:repackcorrectnesstriangles}
\end{lemma}
\begin{proof}
The algorithm packs a set of circles $C$ into a $b$-semihat with capacity $a$. The left child will be a $b$-$b_L$-semihat, for some $b_L \geq 0$, with capacity $a_L$. The right child will be a $b_R$-semihat, for some $b_R \geq 0$, with capacity $a_R$. We do the proof for each of the four cases of the repacking algorithm.

\noindent
To prove statement (1), by Lemma \ref{lem:splitpackingsemihats}, we only need to show that:
\begin{itemize}
\item If the $b$-semihat was instead a full triangle, the circles $C_L$ can be packed into the left child with the Split Packing Algorithm.
\item The largest circle in $C_L$ has size at least $b$.
\end{itemize}
\noindent To prove statement (2), we show that \Call{totalSize}{$C_R$}$\leq a_R$ and that there exists a circle in $C_R$ with size at least $b_R$.
\begin{description}
\item \textbf{Case 1:}

Consider an imaginary circle $c'$ of size $a -$\Call{totalSize}{$C$}. Running \Call{Repack}{} on $C$ is equivalent to a Split Packing of $C \cup \{c'\}$ into a $b$-hat (see Algorithm \ref{alg:split}), where the largest circle is placed into the left child.
\begin{enumerate}
\item % (1)
By the Split Packing analogy, the largest circle alone can be packed into the left child using Split Packing. By the invariant, the largest circle also has size of at least $b$.
\item % (2)
It is clear that the circles in $C_R$ have total size at most $a_R$. The right child is a full triangle (as $a_R \leq a_{sh}^*$), so $b_R = 0$.
\end{enumerate}
\item \textbf{Case 2:}
\begin{enumerate}
\item % (1)
If the original $b$-semihat was a full triangle, then the left child is also a full triangle with capacity $a_{lg}^*$. As \Call{totalSize}{$C_L$}$\leq a_{lg}^*$, the Split Packing Algorithm can pack them into the left child. By the invariant, the largest circle, which will be in $C_L$, also has size of at least $b$.
\item % (2)
There are no circles in $C_R$, and the right child is a full triangle ($b_R=0$), so the invariant trivially holds.
\end{enumerate}
\item \textbf{Case 3:}
\begin{enumerate}
\item % (1)
If the original $b$-semihat was a full triangle, then the left child is also a full triangle with capacity \Call{totalSize}{$C_L$}. Thus the Split Packing Algorithm can pack them into the left child.

Due to Case 1, the largest circle will have size $\leq a_{lg}^*$, and so will be the first circle to be placed in $C_L$. By the invariant, this circle has size at least $b$.
\item % (2)
It is clear that \Call{totalSize}{$C_R$}$\leq a_R^* + \delta a$, the capacity of the right child.

Take any circle $c$ in $C_R$. If $c.size \leq \delta a$, then at the point circle $c$ was visited by the greedy algorithm, it would have been added to $C_L$. Thus we must have $c.size > \delta a$.

For $\delta \geq 0$, let $b_{sh}^\delta$ be the expression in terms of $\delta$ defined in Lemma \ref{lem:trianglebsize}.

The capacity of the right child is $a_R^* + \delta a$, so $b_R = b_{sh}^\delta$. By Lemma \ref{lem:trianglepackinglemma}, as $\delta \leq \delta_{sh}$, we have $c.size > \delta a \geq b_R$, so the invariant holds.
\end{enumerate}
\item \textbf{Case 4:}
($\delta_{sh}$ is defined in Lemma \ref{lem:trianglepackinglemma})

We first show that the two largest circles have sizes in $[\delta_{sh}a, a_{lg}^*-\delta_{sh}a]$.
\begin{description}
\item
Take the largest circle $c \in C$. $c.size \leq a_{lg}^*$ or Case 1 would apply. So $c.size \leq a_{lg}^*-\delta_{sh}a$ otherwise it would be greedily packed into $L$, and Case 3 would apply. Thus all circles in $C$ have sizes at most $a_{lg}^*-\delta_{sh}a$.

$C_R$ has at least one circle, or Case 2 would apply. Any circle in $C_R$ must have size at least $\delta a$ or the greedy algorithm would have packed it into $L$. As the largest circle is in $L$ by the greedy algorithm, there are at least two circles with size at least $\delta a \geq \delta_{sh} a$. Thus the two largest circles must have sizes in $[\delta_{sh}a, a_{lg}^*-\delta_{sh}a]$.
\end{description}
Now let the largest two circles be $c_1, c_2$. We show that $c_1.size + c_2.size > a_{lg}^*$.
\begin{description}
\item
By Lemma \ref{lem:threedeltash}, $3\delta_{sh} \geq a_{lg}^*/a$ for all $s \geq 1$. Thus $c_1.size + c_2.size \geq 2_{sh}a \geq a_{lg}^* - \delta_{sh} a$.
If $c_1.size + c_2.size \leq a_{lg}^*$, then the greedy algorithm would place these two circles into $L$, and $L$ would then have total size at least $a_{lg}^* - \delta_{sh}a$, so Case 3 would apply.
Thus we must have $c_1.size + c_2.size > a_{lg}^*$.
\end{description}
\begin{enumerate}
\item % (1)
Because $a_L = c_1.size + c_2.size > a_{lg}^*$, the left child will not be a full triangle. We show that these two circles fit into the left child (which is the long child).

As both $c_1.size$, $c_2.size \leq a_{lg}^* - \delta_{sh}a$, we have $c_1.size + c_2.size \leq 2a_{lg}^* - 2\delta_{sh}a = a_{lg}^* + (a_{lg}^*/a - 2\delta_{sh})a$.

Let $b_{lg}^\delta$ be the expression in terms of $\delta \geq 0$ defined in Lemma \ref{lem:trianglebsize}. As $b_{lg}^\delta$ is monotonically increasing with $\delta$, we have $b_L \leq b_{lg}^{a_{lg}^*/a-2\delta_{sh}}$. (subbing $\delta := a_{lg}^*/a - 2\delta_{sh}$).

As $c_1.size \in [\delta_{sh}a, a_{lg}^*-\delta_{sh}a]$, we must have $\delta_{sh}a \leq a_{lg}^*-\delta_{sh}a$. Thus $2\delta_{sh} \leq a_{lg}^*/a$, so by Lemma \ref{lem:case4inequality}, we have $\delta_{sh} \geq b_{lg}^{a_{lg}^*/a-2\delta_{sh}}/a$.

Thus $c_1.size$, $c_2.size \geq \delta_{sh}a \geq b_{lg}^{a_{lg}^*/a-2\delta_{sh}} \geq b_L$, so the Split Packing Algorithm can pack them into the left child.

By the invariant, the largest circle, which will be in $C_L$, has size of at least $b$.
\item % (2)
It is clear that the circles in $C_R$ have total size at most $a_R$. As $a_L = c_1.size + c_2.size > a_{lg}^*$, we must have $a_R < a_{sh}^*$, the right child is a full triangle, so $b_R = 0$.
\end{enumerate}
\end{description}
\end{proof}
\noindent The following lemmas were used in the proof in Case 4. Their proofs are in Appendix A.
\begin{lemma}
$3\delta_{sh} \geq a_{lg}^*/a$ for all $s \geq 1$.
\label{lem:threedeltash}
\end{lemma}
\begin{proof}
Appendix \ref{proof:threedeltash}
\end{proof}

\begin{lemma}
If $2\delta_{sh} \leq a_{lg}^*/a$, then we have $\delta_{sh} > b_{lg}^{a_{lg}^*/a - 2\delta_{sh}}/a$, where $b_{lg}^{a_{lg}^*/a - 2\delta_{sh}}$ is defined in Lemma \ref{lem:trianglebsize} for the long child (by letting $\delta := a_{lg}^*/a - 2\delta_{sh}$).
\label{lem:case4inequality}
\end{lemma}
\begin{proof}
Appendix \ref{proof:case4inequality}
\end{proof}

\begin{lemma} [Properties of \Call{Repack}{} on Squares]
Assume that the Packing Invariant (Definition \ref{dfn:packinginvariant}) holds on a square with a set of circles $C$. Suppose that we call \Call{Repack}{$C$,$S$}. Let $C_R \subseteq C$ be the set of circles that \Call{Repack}{$C$,$S$} will pack into the right child of $S$. Then,
\begin{enumerate}
\item All circles in the left child of $S$ will be packed within the bounds of the left child, and no two of these circles will overlap.
\item The packing invariant holds on the right child of $S$ with the circles $C_R$ (i.e. it holds when \Call{Repack}{} recurses in to the right child)

%\item The algorithm packs all the circles in the left child within its bounds.
%\item The packing invariant (Definition \ref{dfn:packinginvariant}) is maintained when recursing into the right child.
\end{enumerate}
\label{lem:repackcorrectnesssquares}
\end{lemma}
\begin{proof}
The algorithm packs a set of circles $C$ into a square with capacity $a$. The left child will be a $b_L$-hat of capacity $a_L$, for some $b_L \geq 0$, while the right child will be a full triangle of capacity $a_R$. We do the proof for each of the two cases of the packing algorithm.

\noindent
To prove (1), we note that the Split Packing Algorithm can pack a set of circles $C$ into a $b_L$-hat of capacity $a_L$ if $\Call{totalSize}{C_L} \leq a_L$ and every circle in $C_L$ has size at least $b_L$.

\noindent
To prove (2), we show that $\Call{totalSize}{C_R} \leq a_R$ and that there exists a circle in $C_R$ with size at least $b_R$.

\begin{description}
\item \textbf{Case 1:}
\begin{enumerate}
\item % (1)
We have $\Call{totalSize}{C_L} = a_L$.
Take any $c \in C_L$. If $c.size \leq \delta a$, then as $\Call{totalSize}{C_L} = a(0.5+\delta)$, at the point $c$ is visited in Algorithm \ref{alg:squarecase1packing}, $c$ would have been moved to $C_R$. Thus we can say that $c.size > \delta a$. Let $b_{sq}^\delta$ be the expression defined in terms of $\delta \geq 0$ in Lemma \ref{lem:squarebsize}. We have $b_L = b_{sq}^\delta$. $\Call{totalSize}{C_L}$ cannot exceed $a$, so we always have $\delta \leq 0.5$. Thus by the Square Packing Lemma (Lemma \ref{lem:squarepackinglemma}), we have $\delta a \geq b_L$, so $c.size > b_L$.
\item % (2)
It is clear that $\Call{totalSize}{C_R} \leq a_R$. The right child is a full triangle, so $b_R = 0$.
\end{enumerate}
\item \textbf{Case 2:}
\begin{enumerate}
\item % (1)
It is clear that $\Call{totalSize}{C_L} \leq a_L$. The left child is a full triangle, so $b_L = 0$.
\item % (2)
There are no circles in $C_R$, and the right child is a full triangle ($b_R=0$), so the invariant trivially holds.
\end{enumerate}
\end{description}
\end{proof}

\begin{proof}[Proof of Theorem \ref{thm:correctness}]
In the binary tree's initial configuration, every node in the tree is validly packed (Definition \ref{dfn:validpacking}) as there are no circles in any of the nodes. Similarly, the Packing Invariant (Definition \ref{dfn:packinginvariant}) holds on the initial empty packing described in Section \ref{section:onlinesplitpackingalgorithm}.

The binary tree can only be modified through calls to the \Call{Repack}{} Algorithm. \Call{Repack}{$C$,$S$} can only be called once per insertion, through Algorithm \ref{alg:insertionalgorithm}. Inductively assuming that the Packing Invariant held on the packing after the previous insertion, by Lemma \ref{lem:insertionrepackcapacity}, \Call{Repack}{$C$,$S$} will only be called when the Packing Invariant (Definition \ref{dfn:packinginvariant}) holds on $S$ with the circles $C$. By then applying Lemmas \ref{lem:repackcorrectnesstriangles} and \ref{lem:repackcorrectnesssquares} inductively, the statements outlined in these two theorems will hold on all rightmost nodes $S$ and below. 

One consequence of this is that the Packing Invariant will now hold on $S$ and all its descendant rightmost nodes. For each ancestor $S'$ of $S$, statement 1 of the Packing Invariant continues to hold as $S'$ is not changed, statement 2 holds by Lemma \ref{lem:insertionrepackcapacity}, and statement 3 continues to hold as no circles are removed from its set of contained circles. Thus, the Packing Invariant continues to hold throughout the packing.

We can then show inductively that every node is validly packed. By Theorem \ref{prop:repacktermination}, every circle is packed at some finite height. Let $H$ be the largest such height over all circles currently packed. All nodes at heights greater than $H$ are validly packed as there are no circles contained in these nodes.
Assume inductively that all nodes at heights greater than $h$ are validly packed, and consider any node $N$ at height $h$.

If $N$ is not a rightmost node, then $N$ will be the left child of some node. $N$ cannot be $S$.
If $N$ is a descendant of $S$, then $N$ is validly packed by Lemmas \ref{lem:repackcorrectnesstriangles} and \ref{lem:repackcorrectnesssquares}.
If $N$ is not a descendant of $S$, then $N$ will be unchanged by the repack, and thus remains validly packed.

If $N$ is a rightmost node,
By Lemma \ref{lem:childrenwithinbounds}, the left and right children are within the bounds of N and their bounds do not overlap.
Both children are packed validly by the induction hypothesis, thus $N$ is also validly packed.

Inductively, this shows that all nodes $N$, up to and including the root, are validly packed.

% IDEA: MERGE WITH EARLIER 

%Now we need to prove inductively that 1. and 2. hold. (separate proof)

%Induction Proof:
%Every circle is packed into the left child of some node at some finite height. Let H be the largest such height over all circles currently packed.
%At all nodes below height H, the induction hypothesis vacuously holds as there are no circles contained within these nodes.

%Assume that all nodes below some height h are validly packed.
%Consider any node N at height h.
%(There are two nodes at height h, one is a rightmost node, and one isn't.)
%If N isn't a rightmost node, the node cannot be S.
%    If the node is a descendant of S, by Theorem 17/20, the node is validly packed.
%    If the node is not a descendant of S, the node will not have been touched by the repack, and thus remains validly packed.
%If N is a rightmost node,
%    We first show that the bounds of the left and right children are within the bounds of N and the bounds of the children do not overlap.
%        If N is S or one of its descendants, this is true by theorem 17/20.
%        If N is not a descendant of S, the bounds of the left and right children will not have been changed by repacking S (this is true even if S is one of its children, as repacking S does not change its bounds).
    
%    Thus, as both children are packed validly by the induction hypothesis, N is also packed validly.

\end{proof}

%S~CROP_END - Proof of Correctness
% SPLIT-----------

\section{Proof of Cost Bound}
\label{section:proofofcostbound}
In this section, we show an amortised reallocation cost of $O(c(1+s^2)\log_{1+s^2}(\frac{1}{c}))$ when inserting a circle of size $c$ into an $s$-shape. If the region is a square, the cost bound becomes $O(c\log_2(\frac{1}{c}))$. Note that $s$ depends only on the shape of the root node, as the children of $s$-shapes are also $s$-shapes. \longshort{}{We first state some results that are shown in the full version of the paper.}

%S~>INS \begin{lemma}
%S~>INS \label{lem:shapeinvariantholds}
%S~>INS For any slack shape, the capacity of the its right child is $\frac{1}{1+s^2}$ of its capacity.
%S~>INS \end{lemma}

%S~>INS \begin{lemma}
%S~>INS When inserting a circle of size $c$ into an $s$-shape of capacity $a$, Algorithm \ref{alg:insertionalgorithm} passes through at most $\floor{\log_{1+s^2} \frac{a}{c}} + 1$ slack shapes before terminating by calling \Call{Repack}{}.
%S~>INS \label{lem:maxslackshapes}
%S~>INS \end{lemma}

\begin{lemma}
Whenever the \Call{Repack}{} algorithm instantiates a new slack shape $S$, its right child will be empty (contains no circles within its bounds).
\label{lem:emptyrightchilds}
\end{lemma}

\begin{proof}
Slack shapes are only produced by Case 2 for \Call{Repack}{} on both triangles and squares. Both cases produce empty right children.
\end{proof}

Algorithm \ref{alg:insertionalgorithm} only repacks slack shapes. Thus, we define a potential function that allocates potential only to slack shapes. As we use circle areas as our cost metric, the reallocation cost of repacking a shape $S$ is at most the capacity of $S$. When a slack shape is newly instantiated after a repack, the right child is initialised as empty (Lemma \ref{lem:emptyrightchilds}). When the shape is to be repacked (when \Call{Repack}{$C$,$S$} is called), the right child, including the newly-inserted circle, would be over capacity. As the capacity of the right child of a slack shape is always exactly $\frac{1}{1+s^2}$ of its parent's capacity (Lemma \ref{lem:shapeinvariantholds}), we define the potential allocated to a slack shape as $(1+s^2)\times\Call{totalSize}{rightChild}$.

With this amount of potential, the shape only needs to draw potential from itself to repack itself. Immediately after a repack, the shape, as well as all its descendants, will store $0$ potential. This is as all newly instantiated slack shapes have empty right children (Lemma \ref{lem:emptyrightchilds}). A newly-inserted circle only contributes to the potential of the slack shapes it passes through, including the one it eventually repacks. As a newly-inserted circle of size $c$ only passes through at most $\floor{\log_{1+s^2} \frac{a}{c}} + 1$ slack shapes before a repack is called (Lemma \ref{lem:maxslackshapes}), the amortised cost of inserting a circle of size $c$ is $c(1+s^2)(\floor{\log_{1+s^2} \frac{a}{c}} + 1)$. Thus we have the following result (Theorem \ref{thm:insertionsamortisedcost}):

\begin{theorem}
\label{thm:insertionsamortisedcost}
In the case of insertions only, using the Online Split Packing Algorithm, we obtain an amortised cost of $O(c(1+s^2)\log_{1+s^2} \frac{1}{c})$ to insert a circle of size $c$ into an $s$-shape.
\end{theorem}
% SPLIT-----------

\section{Online Split Packing for Insertions and Deletions}
\label{section:insertionsanddeletions}
Given any tight online packing algorithm $A$ for insertions only, there is a simple way to extend it to a packing algorithm where both insertions and arbitrary deletions are allowed, by allowing an arbitrarily small amount of slack space. For any fixed $\epsilon \in (0,1)$, the algorithm will achieve a packing density of $a(1-\epsilon)$ on a shape of capacity $a$.

\begin{algorithm}
\caption{Insertions and Deletions with Slack}\label{alg:insertanddeletewithslack}
\begin{algorithmic}[1]
\Procedure{Insert}{$x,S$}
    \If {$x.size$ + \Call{totalSize}{$S$} $\geq S.capacity$}
\State  Remove all inactive objects in $S$, and repack all the active objects into the positions they would be in if they were inserted one-by-one via algorithm $A$'s \Call{Insert}{} operation.
    \Else
\State  Use algorithm $A$'s \Call{Insert}{} operation to insert $x$ into $S$, and mark $x$ as active.
    \EndIf
\EndProcedure
\Procedure{Delete}{$x,S$}
\State  Mark $x$ as inactive.
\EndProcedure
\end{algorithmic}
\end{algorithm}

The basic idea is to perform deletions lazily, only actually removing deleted items through repacking when we run out of space. For the \Call{Insert}{} operation in Algorithm \ref{alg:insertanddeletewithslack}, removing and repacking all the active objects in $S$ has a reallocation cost at most the capacity of $S$. After the repack, the total size of $S$ is at most $(1-\epsilon)\times S.capacity$, so between repacks, at least $\epsilon \times S.capacity$ of insertions must have been done. Thus we need an additional amortised cost of $c \times \frac{1}{\epsilon}$ when inserting an object of size $c$. Suppose that the original insertion algorithm has an amortised reallocation cost of $O(f(c))$ when inserting an object of size $c$. When we allow deletions using Algorithm \ref{alg:insertanddeletewithslack}, we then have an amortised cost of $O(f(c) + c\frac{1}{\epsilon})$ per insertion. We note that Algorithm \ref{alg:insertanddeletewithslack} does not need to know the value of $\epsilon$ being used.

Applying this method in the context of circle packing with the earlier described Online Split Packing Algorithm for insertions only, we obtain the following result (Theorem \ref{thm:insertdeleteamortisedcost}):

\begin{theorem}
\label{thm:insertdeleteamortisedcost}
When allowing both insertions and deletions into an $s$-shape of capacity $a$, for any fixed $\epsilon > 0$, the Online Split Packing Algorithm achieves a packing density of $(1-\epsilon)a$, with an amortised reallocation cost of $O(c((1+s^2)\log_{1+s^2} \frac{1}{c} + \frac{1}{\epsilon}))$ for inserting a circle of area $c$. More specifically, for the insertion and deletion of circles into a square, we obtain an amortised reallocation cost of $O(c(\log_2 \frac{1}{c} + \frac{1}{\epsilon}))$.
\end{theorem}

\section{Conclusion} %S~<DEL
We have adapted the Split Packing Algorithm to handle an online sequence of insertions and deletions and pack arbitrarily close to critical density by allowing reallocations. Our cost bound is asymptotically equal across $s$-shapes with differing values of $s$. %S~<DEL

The Split Packing algorithm has also been shown to work for packing shapes other than circles, like squares and octagons, into triangles. The paper goes on to define a new type of shape, a ``Gem'', which represents the most general type of shape that the Split Packing Algorithm can handle. Due to the close relationship between the algorithms, it is likely that a similar of generalisation would apply to the Online Split Packing Algorithm. %S~<DEL

We note that the deletion procedure defined in Section \ref{section:insertionsanddeletions} works for any tight packing algorithm for insertions only into a fixed space. The deletion procedure can also be shown to work not only for tight packing algorithms, but also for packing algorithms that achieve packing densities that are arbitrarily close to the critical density. %S~<DEL

The current reallocation cost bounds apply only to volume costs. A possible direction of future work would be to understand how the cost bounds differ for the other cost models like unit cost (constant cost for each circle reallocation). %S~<DEL

\appendix
\section{Appendix: Proofs of Lemmas}

This appendix contains the proofs which have been omitted from the main description of the Online Split Packing Algorithm in the paper.

\begin{figure}[!h]
  \centering
    \begin{subfigure}[b]{.5\linewidth}
      \centering
      \includegraphics[width=.9\linewidth]{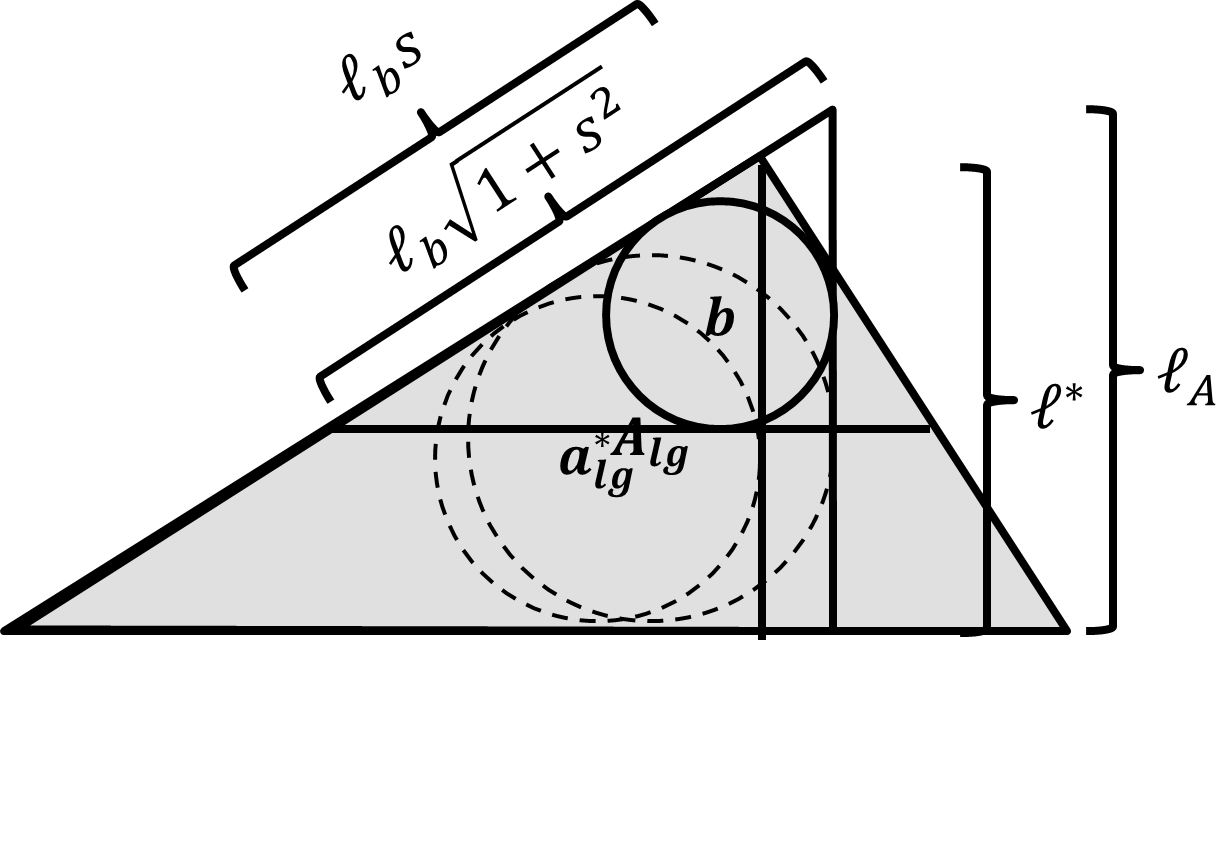}
      \caption{Long Side}
      \label{fig:trianglelongsidemeasurements}
    \end{subfigure}%
    \begin{subfigure}[b]{.5\linewidth}
      \centering
      \includegraphics[width=.9\linewidth]{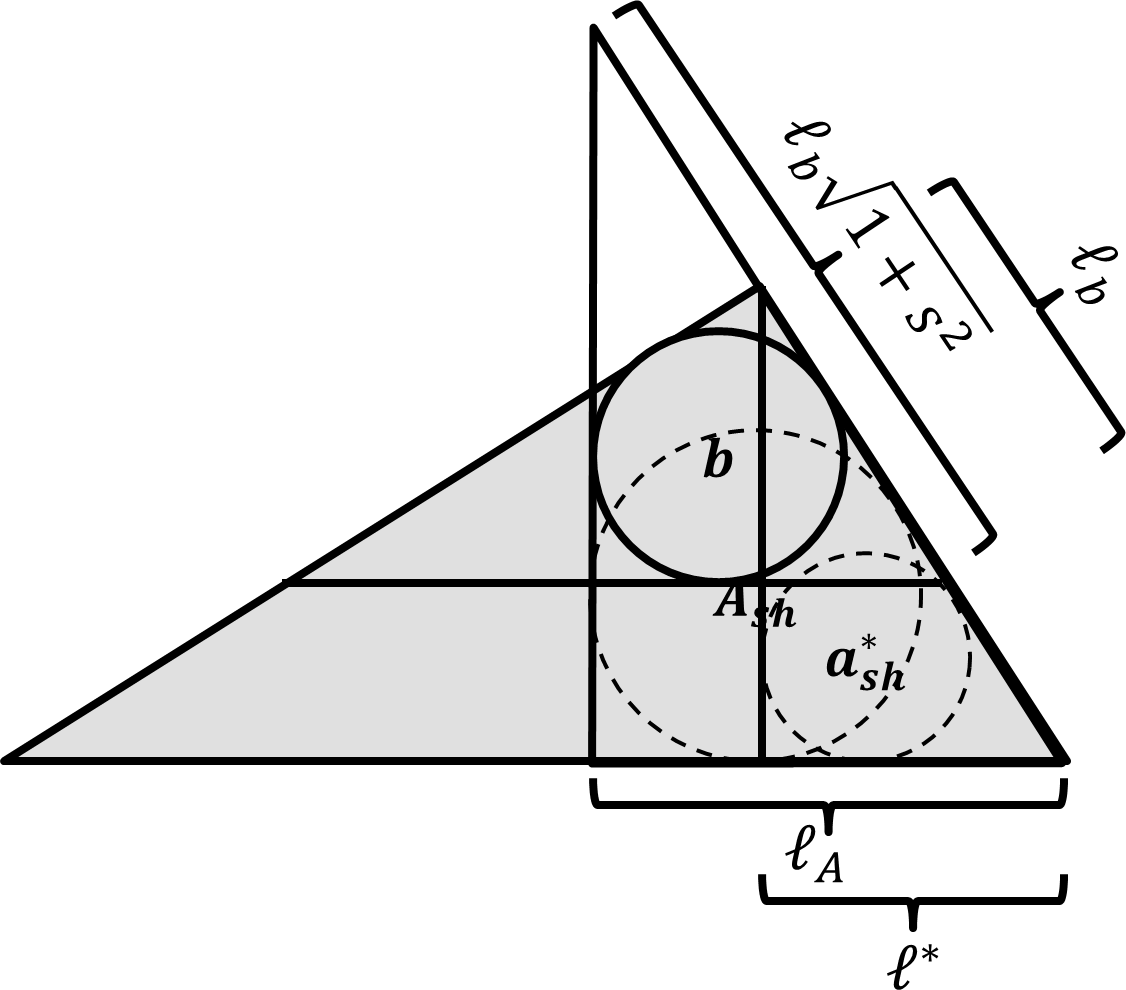}
      \caption{Short Side}
      \label{fig:triangleshortsidemeasurements}
    \end{subfigure}%
  \caption{Measurements for computing a triangle's $b$-curve.}
  \label{fig:measurementsfortriangles}
\end{figure}

\subsection{Lemma \protect\ref{lem:trianglebsize}: Triangle $b$-curve size}
\label{proof:trianglebsize}
\begin{proof}
~\\ Suppose that the bounds causes a $b$-curve to be formed for some $b \geq 0$.

\noindent For the case of the long child, we have three similar triangles, corresponding to the incircle areas of $a_{lg}^*$, $a_{lg}$ and $b$ respectively. In the case of the short child, the incircles are of areas $a_{sh}^*$, $a_{sh}$ and $b$ instead.

\noindent Let the length of the shortest sides of these triangles be $\ell^*$, $\ell_A$ and $\ell_b$ respectively.

\noindent \textbf{Long Child:} (Note $a_{lg} := a_{lg}^* + \delta a$, for some $\delta \geq 0$)
\\ From Figure \ref{fig:trianglelongsidemeasurements}, by computing the length of the part of the shape extending out of the original triangle, we get:
\\ $\ell_A\sqrt{1+s^2} - \ell^*\sqrt{1+s^2} = \ell_b\sqrt{1+s^2}-\ell_b\times s$
\\ $\ds (\frac{\ell_A}{\ell^*} - 1)\Big(\frac{\sqrt{1+s^2}}{\sqrt{1+s^2} - s}\Big) = \frac{\ell_b}{\ell_*}$
\\ $\ds \Big(\sqrt{\frac{a_{lg}}{a_{lg}^*}} - 1\Big)\Big(\frac{\sqrt{1+s^2}}{\sqrt{1+s^2} - s}\Big) = \sqrt{\frac{b}{a_{lg}^*}}$
\\ $\ds b = \Big(\frac{(\sqrt{1+s^2})(\sqrt{a_{lg}} - \sqrt{a_{lg}^*})}{\sqrt{1+s^2} - s}\Big)^2$
\\ $\ds b = a\Big(\frac{(\sqrt{1+s^2})\big(\sqrt{a_{lg}/a} - \sqrt{a_{lg}^*/a}\big)}{\sqrt{1+s^2} - s}\Big)^2$
\\ $\ds b = a\Big(\frac{(\sqrt{1+s^2})\big(\sqrt{a_{lg}^*/a + \delta} - \sqrt{a_{lg}^*/a}\big)}{\sqrt{1+s^2} - s}\Big)^2$
\\ $\ds b = a\Big(\frac{(\sqrt{1+s^2})\big(\sqrt{\frac{s^2}{1+s^2} + \delta} - \frac{s}{\sqrt{1+s^2}}\big)}{\sqrt{1+s^2} - s}\Big)^2$
\\ $\ds b = a\Big(\frac{\sqrt{s^2 + \delta(1+s^2)} - s}{\sqrt{1+s^2} - s}\Big)^2$
\\ 
\\ \textbf{Short Child:} (Note $a_{sh} := a_{sh}^* + \delta a$, for some $\delta \geq 0$)
\\ From Figure \ref{fig:triangleshortsidemeasurements}, by computing the length of the part of the shape extending out of the original triangle, we get:
\\ $\ell_A\sqrt{1+s^2} - \ell^*\sqrt{1+s^2} = \ell_b\sqrt{1+s^2}-\ell_b$
\\ $\ds (\frac{\ell_A}{\ell^*} - 1)\Big(\frac{\sqrt{1+s^2}}{\sqrt{1+s^2} - 1}\Big) = \frac{\ell_b}{\ell_*}$
\\ $\ds \Big(\sqrt{\frac{a_{sh}}{a_{sh}^*}} - 1\Big)\Big(\frac{\sqrt{1+s^2}}{\sqrt{1+s^2} - 1}\Big) = \sqrt{\frac{b}{a_{sh}^*}}$
\\ $\ds b = \Big(\frac{(\sqrt{1+s^2})(\sqrt{a_{sh}} - \sqrt{a_{sh}^*})}{\sqrt{1+s^2} - 1}\Big)^2$
\\ $\ds b = a\Big(\frac{(\sqrt{1+s^2})\big(\sqrt{a_{sh}/a} - \sqrt{a_{sh}^*/a}\big)}{\sqrt{1+s^2} - 1}\Big)^2$
\\ $\ds b = a\Big(\frac{(\sqrt{1+s^2})\big(\sqrt{a_{sh}^*/a + \delta} - \sqrt{a_{sh}^*/a}\big)}{\sqrt{1+s^2} - 1}\Big)^2$
\\ $\ds b = a\Big(\frac{(\sqrt{1+s^2})\big(\sqrt{\frac{1}{1+s^2} + \delta} - \frac{1}{\sqrt{1+s^2}}\big)}{\sqrt{1+s^2} - 1}\Big)^2$
\\ $\ds b = a\Big(\frac{\sqrt{1 + \delta(1+s^2)} - 1}{\sqrt{1+s^2} - 1}\Big)^2$
\end{proof}

\subsection{Lemma \protect\ref{lem:trianglepackinglemma}: Triangle Packing Lemma}
\label{proof:trianglepackinglemma}
\begin{proof}
~\\ \textbf{Short Child:}
\\ $\displaystyle b_{sh}^\delta/a = \Big(\frac{\sqrt{1+\delta(1+s^2)} - 1}{\sqrt{1+s^2} - 1}\Big)^2$
\\ Define $y:=\sqrt{1+s^2}$, and $h:= \sqrt{1+\delta(1+s^2)}$, we have $\delta = \frac{h^2-1}{1+s^2} = \frac{h^2-1}{y^2}$
\\ For $\delta > b_{sh}^\delta/a$, we need $\delta > \Big(\frac{\sqrt{1+\delta(1+s^2)} - 1}{\sqrt{1+s^2} - 1}\Big)^2$
\\ $\iff \frac{h^2-1}{y^2} > \Big(\frac{h-1}{y-1}\Big)^2$
\\ $\iff (h^2-1)(y-1)^2 > (h-1)^2y^2$ as $y > 1$
\\ $\iff (h^2-1)(y^2-2y+1) > (h^2-2h+1)y^2$ as $y > 1$
\\ $\iff h^2(y^2-2y+1) + (-y^2+2y-1) > h^2y^2 - 2hy^2 + y^2$
\\ $\iff h^2(-2y+1) + h(2y^2) + (-2y^2+2y-1) > 0$
\\ As $y>1$, we have $-2y+1 < 0$, so the quadratic curve is concave downwards.
\\ Using the quadratic formula, we obtain the following expression for the roots:
\\ $\ds h = \frac{-2y^2 \pm \sqrt{4y^4 - 4(-2y+1)(-2y^2+2y-1)}}{2(-2y+1)}$
\\ $\ds = \frac{-2y^2 \pm 2\sqrt{y^4 - 4y^3 + 6y^2 - 4y + 1}}{2(-2y+1)}$
\\ $\ds = \frac{-y^2 \pm \sqrt{(y-1)^4}}{-2y+1}$
\\ $\ds = \frac{-y^2 \pm (y-1)^2}{-2y+1}$
\\ $\ds = \frac{-y^2 \pm (y^2-2y+1)}{-2y+1}$
\\ 
\\ As $(y-1)^2 \geq 0$ and $-2y+1 < 0$, the expression is positive for $\ds \frac{-2y+1}{-2y+1} \leq h \leq \frac{-2y^2 + 2y - 1}{-2y+1}$
\\ which is equivalent to:
\\ $\ds 1 \leq h \leq \frac{1}{2}\big(2y-1 + \frac{1}{2y-1}\big)$
\\ $\iff 1 \leq \sqrt{1+\delta y^2} \leq \frac{1}{2}\big(2y-1 + \frac{1}{2y-1}\big)$
\\ $\iff \frac{0}{y^2} \leq \delta \leq \frac{1}{y^2}\big(\frac{1}{4}\big((2y-1) + \frac{1}{2y-1}\big)^2 - 1\big)$
\\ $\iff 0 \leq \delta \leq \frac{1}{y^2}\big(\frac{1}{4}\big((2y-1)^2 + 2 + \frac{1}{(2y-1)^2}\big) - 1\big)$
\\ $\iff 0 \leq \delta \leq \frac{1}{y^2}\big(\frac{1}{4}\big((2y-1)^2 - 2 + \frac{1}{(2y-1)^2}\big)\big)$
\\ $\iff 0 \leq \delta \leq \frac{1}{4y^2}\big((2y-1) - \frac{1}{2y-1}\big)^2$
\\ $\iff 0 \leq \delta \leq \big(\frac{(2y-1)^2 - 1}{2y(2y-1)}\big)^2$
\\ $\iff 0 \leq \delta \leq \big(\frac{4y^2-4y}{4y^2-2y}\big)^2$
\\ $\iff 0 \leq \delta \leq \big(\frac{2y-2}{2y-1}\big)^2$
\\ $\iff 0 \leq \delta \leq \big(1 - \frac{1}{2y-1}\big)^2$
\\ $\iff 0 \leq \delta \leq \big(1 - \frac{1}{2\sqrt{1+s^2}-1}\big)^2$
\end{proof}

\begin{figure}[!h]
  \centering
    \begin{subfigure}[b]{.4\linewidth}
      \centering
      \includegraphics[width=.8\linewidth]{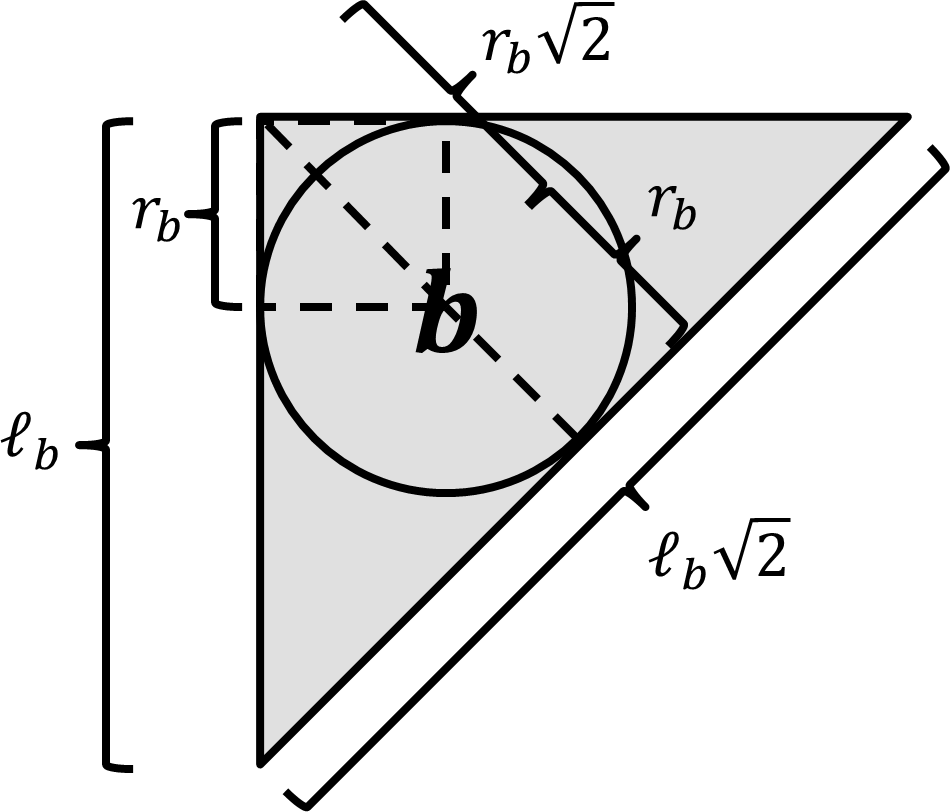}
      \caption{Computing $r_b$ in terms of $\ell_b$}
      \label{fig:squareincirclemeasurements}
    \end{subfigure}%
    \begin{subfigure}[b]{.6\linewidth}
      \centering
      \includegraphics[width=.8\linewidth]{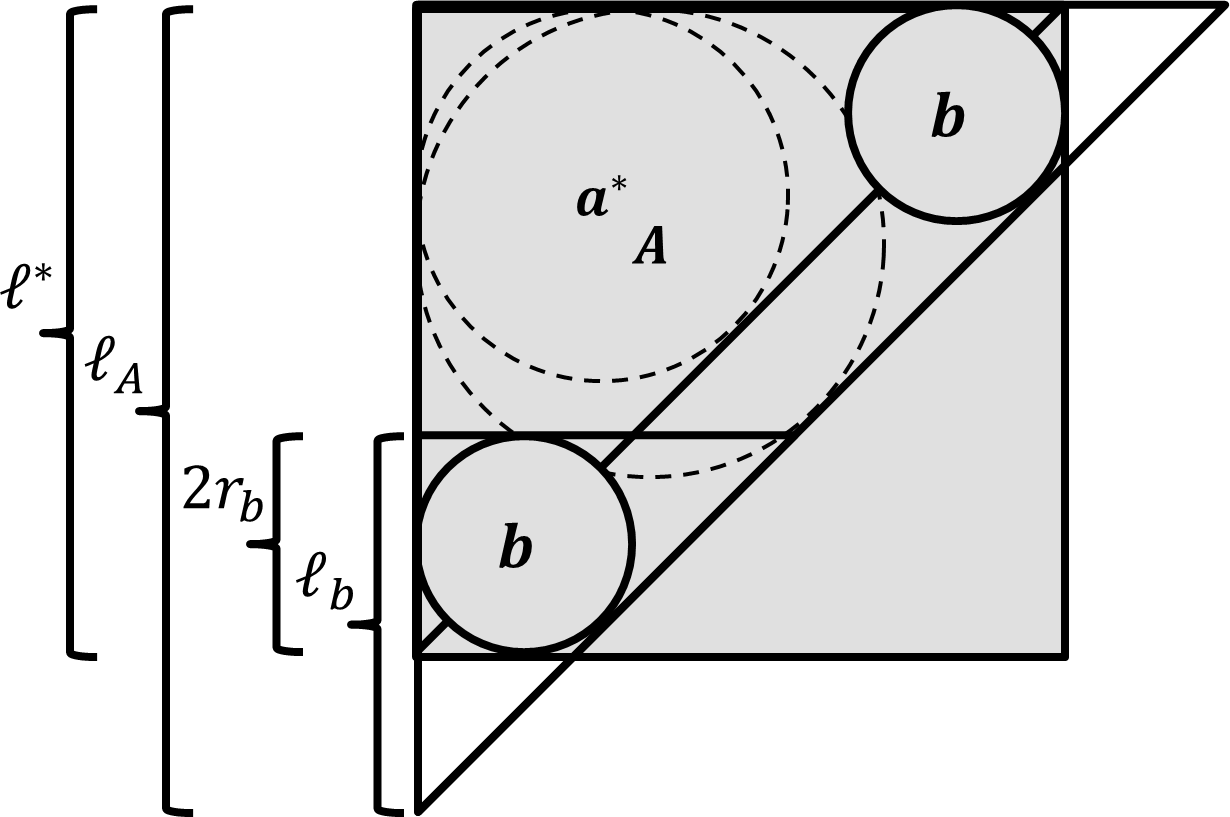}
      \caption{Splitting a Square}
      \label{fig:squaremeasurements}
    \end{subfigure}%
  \caption{Measurements for computing a square's $b$-curve.}
  \label{fig:measurementsforsquares}
\end{figure}

\subsection{Lemma \protect\ref{lem:squarebsize}: Square $b$-curve size}
\label{proof:squarebsize}
\begin{proof}
~\\ Suppose that the bounds causes a $b$-curve to be formed for some $b \geq 0$.

\noindent We have three similar triangles, corresponding to the incircle areas of $a^* := a/2$, $A$ and $b$ respectively. Let the length of the shortest sides of these triangles be $\ell^*$, $\ell_A$ and $\ell_b$ respectively. Let $r_b$ be the radius of a circle of area $b$.

\noindent First, we compute $r_b$ in terms of $\ell_b$ for a square. From Figure \ref{fig:squareincirclemeasurements},
\\ $r_b\sqrt{2} + r_b = \frac{1}{2}\ell_b \sqrt{2}$
\\ $r_b = \frac{1}{\sqrt{2}}\times \frac{1}{\sqrt{2}+1} \ell_b$
\\ $= (1 - \frac{1}{\sqrt{2}})\ell_b$

\noindent (Note $A := a^* + \delta a$, for some $\delta \geq 0$)
\\ From Figure \ref{fig:squaremeasurements}, by computing the length of the part of the shape extending out of the original square, we get:
\\ $\ell_A - \ell^* = \ell_b - 2r_b$
\\ $\ell_A - \ell^* = \ell_b - 2(1 - \frac{1}{\sqrt{2}})\ell_b$
\\ $\ell_A - \ell^* = \ell_b(\sqrt{2}-1)$
\\ $\ds \frac{\ell_A}{\ell^*} - 1 = \frac{\ell_b}{\ell^*}(\sqrt{2}-1)$
\\ $\ds \sqrt{\frac{A}{a^*}} - 1 = \sqrt{\frac{b}{a^*}}(\sqrt{2}-1)$
\\ $\ds \sqrt{\frac{A}{a^*}} - 1 = \sqrt{\frac{b}{a^*}}(\sqrt{2}-1)$
\\ $\ds b = \Big(\frac{\sqrt{A} - \sqrt{a^*}}{\sqrt{2} - 1}\Big)^2$
\\ $\ds b = a\Big(\frac{\sqrt{A/a} - \sqrt{a^*/a}}{\sqrt{2} - 1}\Big)^2$
\\ $\ds b = a\Big(\frac{\sqrt{0.5 + \delta} - \sqrt{0.5}}{\sqrt{2} - 1}\Big)^2$
\\ $\ds b = a\Big(\frac{\sqrt{1 + 2\delta} - 1}{2 - \sqrt{2}}\Big)^2$
\end{proof}

\subsection{Lemma \protect\ref{lem:squarepackinglemma}: Square Packing Lemma}
\label{proof:squarepackinglemma}
\begin{proof}
~\\ Letting $h := \sqrt{1+2\delta}$, we have $\delta = \frac{h^2-1}{2}$
\\ $\ds \delta \geq b_{sq}^\delta/a = \Big(\frac{\sqrt{1+2\delta} - 1}{\sqrt{2} - 2} \Big)^2$
\\ $\ds \iff \frac{h^2-1}{2} \geq \Big(\frac{h-1}{\sqrt{2}-2}\Big)^2$
\\ $\ds \iff (h^2-1)(\sqrt{2}-2)^2 \geq 2(h-1)^2$
\\ $\ds \iff (h^2-1)(6-4\sqrt{2}) \geq 2h^2-4h+2$
\\ $\ds \iff h^2(4-4\sqrt{2}) + 4h + (4\sqrt{2} - 8) \geq 0$
\\ $\ds \iff h^2(1-\sqrt{2}) + h + (\sqrt{2} - 2) \geq 0$

\noindent As $4 - 4\sqrt{2} < 0$, the quadratic curve is concave downwards.
\\ Using the quadratic formula, the roots are:
\\ $\ds h = \frac{-1 \pm \sqrt{1 - 4(1-\sqrt{2})(\sqrt{2}-2)}}{2(1-\sqrt{2})}$
\\ $\ds = \frac{-1 \pm \sqrt{17 - 12\sqrt{2}}}{2(1-\sqrt{2})} = \frac{-1 \pm \sqrt{(3 - 2\sqrt{2})^2}}{2(1-\sqrt{2})} = \frac{-1 \pm (3 - 2\sqrt{2})}{2-2\sqrt{2}}$
\\ Thus the above inequality holds if and only if
\\ $\ds \frac{2 - 2\sqrt{2}}{2-2\sqrt{2}} \leq h \leq \frac{2\sqrt{2} - 4}{2 - 2\sqrt{2}}$
\\ $\iff \ds 1 \leq \sqrt{1+2\delta} \leq \sqrt{2}$
\\ $\iff \ds 0 \leq \delta \leq 0.5$
\end{proof}

\subsection{Lemma \protect\ref{lem:threedeltash}}
\label{proof:threedeltash}
\begin{proof}
~\\ Let $y = \sqrt{1+s^2}$
\\ Thus we have:
\\ $\ds \dsh = \Big(\frac{2y-2}{2y-1}\Big)^2 = 4\Big(\frac{y-1}{2y-1}\Big)^2$
\\ $\ds \alg/a = \frac{y^2-1}{y^2}$

\noindent Therefore,
\\ $\ds 3\delta_{sh} > \alg/a \iff 3\times 4\Big(\frac{y-1}{2y-1}\Big)^2 > \frac{(y+1)(y-1)}{y^2}$
\\ $\iff 12\frac{y-1}{(2y-1)^2} > \frac{y+1}{y^2}$
\\ $\iff 12(y^3-y^2) > (y+1)(2y-1)^2$
\\ $\iff 12(y^3-y^2) > (y+1)(4y^2-4y+1)$
\\ $\iff 12y^3- 12y^2 > 4y^3 - 3y + 1$
\\ $\iff 8y^3- 12y^2 + 3y - 1 > 0$
\\ Which is true for all $y \geq \sqrt{2}$, i.e. for all $s \geq 1$.
\end{proof}

\subsection{Lemma \protect\ref{lem:case4inequality}}
\label{proof:case4inequality}
\begin{proof}
~\\ Suppose $2\dsh \leq \alg/a$

\noindent Then we have $(\alg/a - 2\dsh)(1+s^2) \geq 0$, so
\\ $\sqrt{s^2 + (\alg/a - 2\dsh)(1+s^2)} - s \geq 0$.

\noindent We want to show that $\dsh > b_{lg}^{\alg/a - 2\dsh} / a$ on the long child.
\\ We have:
\\ $\dsh > b_{lg}^{\alg/a - 2\dsh} / a$
\\ $\iff \ds \dsh > \Big(\frac{\sqrt{s^2 + (\alg/a - 2\dsh)(1+s^2)}-s}{\sqrt{1+s^2}-s}\Big)^2$
\\ (by Lemma \ref{lem:trianglebsize}, where $b$ is defined on the long child)
\\ $\iff \ds \Big(1 - \frac{1}{2\sqrt{1+s^2}-1}\Big)^2 > \Big(\frac{\sqrt{s^2 + (\alg/a - 2\dsh)(1+s^2)}-s}{\sqrt{1+s^2}-s}\Big)^2$
\\ 
\\ $\iff \ds \Big(\frac{2\sqrt{1+s^2}-2}{2\sqrt{1+s^2}-1}\Big)^2 > \Big(\frac{\sqrt{s^2 + (\alg/a - 2\dsh)(1+s^2)}-s}{\sqrt{1+s^2}-s}\Big)^2$
\\ 
\\ $\iff \ds \frac{2\sqrt{1+s^2}-2}{2\sqrt{1+s^2}-1} > \frac{\sqrt{s^2 + (\alg/a - 2\dsh)(1+s^2)}-s}{\sqrt{1+s^2}-s}$
\\ (As $s \geq 1$, the numerators and denominators of both sides are nonnegative)

\noindent Letting $y = \sqrt{1+s^2}$, we have

\noindent $\iff \ds 2\frac{y-1}{2y-1} > \frac{\sqrt{2(y^2-1) - 8\big(\frac{y-1}{2y-1}\big)^2 y^2}-s}{y-s}$
\\ $\iff 2(y-1)(y-s) > \Big(\sqrt{2(y^2-1) - 8\big(\frac{y-1}{2y-1}\big)^2 y^2}-s\Big)(2y-1)$
\\ (as the denominators are always positive)

\noindent $\iff 2(y-1)(y-s) > \sqrt{2(2y-1)^2(y^2-1) - 8(y-1)^2 y^2}-s(2y-1)$
\\ $\iff 2(y^2-y-ys+s) + s(2y-1) > \sqrt{8y^3 - 14y^2 + 8y - 2}$
\\ $\iff 2y^2 - 2y + s > \sqrt{8y^3 - 14y^2 + 8y - 2}$

\noindent As $s \geq 1$, it suffices to show that:
\\ $2y^2 - 2y + 1 > \sqrt{8y^3 - 14y^2 + 8y - 2}$
\\ As $2y^2 - 2y + 1$ is always nonnegative, we have:
\\ $\iff (2y^2 - 2y + 1)^2 > 8y^3 - 14y^2 + 8y - 2$
\\ $\iff 4y^4 - 16y^3 + 22y^2 - 12y + 3 > 0$
\\ Which is always true as this polynomial has no roots.

\noindent Therefore $2\dsh \leq \alg/a$ implies that $\dsh > b_{lg}^{\alg/a - 2\dsh} / a$ on the long child.
\end{proof}

%%
%% Bibliography
%%

%% Please use bibtex, 
\tiny
%S~>INS \bibliography{report_autotrim}
\bibliography{report}  %S~<DEL

\begin{thebibliography}{10}

\bibitem{ontwodimensionalpacking}
Yossi Azar and Leah Epstein.
\newblock On two dimensional packing.
\newblock {\em Journal of Algorithms}, 25(2):290 -- 310, 1997.

\bibitem{cacheobliviousbtrees}
Michael~A Bender, Erik~D Demaine, and Martin Farach-Colton.
\newblock Cache-oblivious b-trees.
\newblock In {\em Proceedings 41st Annual Symposium on Foundations of Computer
  Science}, pages 399--409. IEEE, 2000.

\bibitem{bender2014cost}
Michael~A. Bender, Martin Farach-Colton, Sandor~P. Fekete, Jeremy~T. Fineman,
  and Seth Gilbert.
\newblock Cost-oblivious storage reallocation.
\newblock In {\em Proceedings of the 33rd ACM SIGMOD-SIGACT-SIGART Symposium on
  Principles of Database Systems}, PODS '14, pages 278--288, New York, NY, USA,
  June 2014. ACM.

\bibitem{bender2015reallocation}
Michael~A Bender, Martin Farach-Colton, S{\'a}ndor~P Fekete, Jeremy~T Fineman,
  and Seth Gilbert.
\newblock Reallocation problems in scheduling.
\newblock {\em Algorithmica}, 73(2):389--409, August 2014.

\bibitem{maintainingarrays}
Michael~A. Bender, S{\'a}ndor~P. Fekete, Tom Kamphans, and Nils Schweer.
\newblock Maintaining arrays of contiguous objects.
\newblock In Miros{\l}aw Kuty{\l}owski, Witold Charatonik, and Maciej
  G{\k{e}}bala, editors, {\em Fundamentals of Computation Theory}, pages
  14--25, Berlin, Heidelberg, 2009. Springer Berlin Heidelberg.

\bibitem{adaptivepma}
Michael~A. Bender and Haodong Hu.
\newblock An adaptive packed-memory array.
\newblock {\em ACM Trans. Database Syst.}, 32(4), November 2007.

\bibitem{berndt2014fully}
Sebastian Berndt, Klaus Jansen, and Kim{-}Manuel Klein.
\newblock Fully dynamic bin packing revisited.
\newblock In {\em APPROX-RANDOM}, volume~40, pages 135--151. Schloss Dagstuhl -
  Leibniz-Zentrum fuer Informatik, 2015.

\bibitem{improvedsquareintosquare}
Brian Brubach.
\newblock Improved bound for online square-into-square packing.
\newblock In Evripidis Bampis and Ola Svensson, editors, {\em Approximation and
  Online Algorithms}, pages 47--58, Cham, 2015. Springer International
  Publishing.

\bibitem{demaine2010circle}
Erik~D. Demaine, S{\'{a}}ndor~P. Fekete, and Robert~J. Lang.
\newblock Circle packing for origami design is hard.
\newblock In {\em Proceedings of the 5th International Conference on Origami in
  Science, Mathematics and Education}, pages 609--626, Singapore, 2010. A K
  Peters/CRC Press.

\bibitem{online_multidim_bin_packing}
Leah Epstein and Rob van Stee.
\newblock Optimal online bounded space multidimensional packing.
\newblock In {\em Proceedings of the Fifteenth Annual ACM-SIAM Symposium on
  Discrete Algorithms}, SODA '04, pages 214--223, 2004.

\bibitem{Epstein2005}
Leah Epstein and Rob van Stee.
\newblock Online square and cube packing.
\newblock {\em Acta Informatica}, 41(9):595--606, Oct 2005.

\bibitem{square_into_square}
S{\'a}ndor~P. Fekete and Hella-Franziska Hoffmann.
\newblock Online square-into-square packing.
\newblock {\em Algorithmica}, 77(3), Mar 2017.

\bibitem{onlinesquarepacking}
S{\'a}ndor~P. Fekete, Tom Kamphans, and Nils Schweer.
\newblock Online square packing.
\newblock In Frank Dehne, Marina Gavrilova, J{\"o}rg-R{\"u}diger Sack, and
  Csaba~D. T{\'o}th, editors, {\em Algorithms and Data Structures}, pages
  302--314, Berlin, Heidelberg, 2009. Springer Berlin Heidelberg.

\bibitem{splitpacking_ws}
S{\'a}ndor~P. Fekete, Sebastian Morr, and Christian Scheffer.
\newblock Split packing: Packing circles into triangles with optimal worst-case
  density.
\newblock In Faith Ellen, Antonina Kolokolova, and J{\"o}rg-R{\"u}diger Sack,
  editors, {\em Algorithms and Data Structures}, pages 373--384, Cham, 2017.
  Springer International Publishing.

\bibitem{fekete2017efficient}
Sándor~P. Fekete, Jan-Marc Reinhardt, and Christian Scheffer.
\newblock An efficient data structure for dynamic two-dimensional
  reconfiguration.
\newblock {\em Journal of Systems Architecture}, 75(Supplement C):15 -- 25,
  April 2017.

\bibitem{george1995packing}
John~A George, Jennifer~M George, and Bruce~W Lamar.
\newblock Packing different-sized circles into a rectangular container.
\newblock {\em European Journal of Operational Research}, 84(3):693--712,
  August 1995.

\bibitem{onlineremovablesquarepacking}
Xin Han, Kazuo Iwama, and Guochuan Zhang.
\newblock Online removable square packing.
\newblock {\em Theory of Computing Systems}, 43(1):38--55, Jul 2008.

\bibitem{circleliteraturereview}
Mhand Hifi and Rym M'Hallah.
\newblock A literature review on circle and sphere packing problems: Models and
  methodologies.
\newblock {\em Advances in Operations Research}, 2009:1--22, 2009.

\bibitem{hokama2016bounded}
Pedro Hokama, Fl{\'a}vio~K Miyazawa, and Rafael~CS Schouery.
\newblock A bounded space algorithm for online circle packing.
\newblock {\em Information Processing Letters}, 116(5):337--342, May 2016.

\bibitem{ivkovic2009fully}
Zoran Ivkovi{\'{c}} and Errol~L. Lloyd.
\newblock Fully dynamic bin packing.
\newblock In {\em Fundamental Problems in Computing}, pages 407--434. Springer
  Netherlands, 2009.

\bibitem{onlinecubes}
Janusz Januszewski and Marek Lassak.
\newblock On-line packing sequences of cubes in the unit cube.
\newblock {\em Geometriae Dedicata}, 67(3):285--293, 1997.

\bibitem{lintzmayer2017online}
Carla~Negri Lintzmayer, Fl{\'{a}}vio~Keidi Miyazawa, and Eduardo~Candido
  Xavier.
\newblock Online circle and sphere packing.
\newblock {\em CoRR}, abs/1708.08906, August 2017.
\newblock \href {http://arxiv.org/abs/1708.08906} {\path{arXiv:1708.08906}}.

\bibitem{locatelli2002packing}
Marco Locatelli and Ulrich Raber.
\newblock Packing equal circles in a square: a deterministic global
  optimization approach.
\newblock {\em Discrete Applied Mathematics}, 122(1):139--166, October 2002.

\bibitem{splitpacking}
Sebastian Morr.
\newblock Split packing: an algorithm for packing circles with up to critical
  density.
\newblock Master's thesis, Technische Universit\"at Braunschweig, Institute of
  Operating Systems and Computer Networks, June 2016.

\bibitem{splitpackingsoda}
Sebastian Morr.
\newblock Split packing: An algorithm for packing circles with optimal
  worst-case density.
\newblock In {\em Proceedings of the Twenty-Eighth Annual {ACM-SIAM} Symposium
  on Discrete Algorithms, {SODA} 2017}, pages 99--109, 2017.

\bibitem{packomania}
Eckard Specht.
\newblock Packomania.
\newblock \protect\url{http://www.packomania.com/}.
\newblock Accessed: 2017-11-06.

\bibitem{spherepacking}
Milos Tatarevic.
\newblock On limits of dense packing of equal spheres in a cube.
\newblock {\em Electr. J. Comb.}, 22:P1.35, 2015.

\bibitem{onespaceboundedtwodimensional}
Yong Zhang, Jingchi Chen, Francis Y.~L. Chin, in~Han, Hing-Fung Ting, and
  Yung~H. Tsin.
\newblock Improved online algorithms for 1-space bounded 2-dimensional bin
  packing.
\newblock In {\em Algorithms and Computation}, pages 242--253, 2010.

\end{thebibliography}

\end{document}